\documentclass[11pt]{article}
\usepackage[utf8]{inputenc}
\usepackage{amsmath}
\usepackage{amssymb}
\usepackage{amsthm}
\usepackage{url}
\usepackage{hyperref}
\usepackage{xcolor}
\usepackage{dsfont}
\usepackage{float}
\usepackage{mathrsfs}
\usepackage{yfonts}
\newtheorem{definition}{Definition}

\newtheorem{remark}[definition]{Remark}
\newtheorem{theorem}[definition]{Theorem}
\newtheorem{lemma}[definition]{Lemma}
\newtheorem{corollary}[definition]{Corollary}

\usepackage{graphicx}
\usepackage{fullpage}
 \usepackage{caption}
 \usepackage{subcaption}
 \usepackage{verbatim}
\usepackage[title]{appendix}
\usepackage{upgreek}

\usepackage{tikz}
\usetikzlibrary{shapes.geometric,plotmarks,backgrounds,fit,calc,circuits.ee.IEC}
\usetikzlibrary{patterns,tikzmark}
\usetikzlibrary{decorations.pathreplacing}

\usetikzlibrary{matrix}
\usepackage{hf-tikz}

\usepackage{pgfplotstable}

\usepackage{epsfig}
\usetikzlibrary{shapes.symbols,patterns} % for source symbols
\usepackage{pgfplots}
\usepgfplotslibrary{ternary}
\pgfplotsset{compat=newest}

\usepackage{quantikz}

\usepackage{ytableau}

\usepackage{hyperref}

\usepackage{makecell}
\usepackage{cancel}

\newtheorem{proposition}{Proposition}

\allowdisplaybreaks[4]

\usepackage{braket}

\newcommand{\mr}{\mathrm}

\newcommand{\id}{\mathds{1}}

\DeclareMathOperator{\tr}{Tr}

\newcommand{\sym}{\mathrm{Sym}}

\newcommand{\tp}{{\tt{t}}} %partial transpose
\newcommand{\ketbra}[2]{|#1\rangle\!\langle #2|}
\let\inner\relax
\NewDocumentCommand\inner{mg}{%
	\ensuremath{\left\langle #1 \middle\vert \IfNoValueTF{#2}{#1}{#2}\right\rangle}%
}

\usepackage{listings}
\usepackage{matlab-prettifier}

\pgfkeys{tikz/mymatrixenv/.style={decoration={brace},every left delimiter/.style={xshift=8pt},every right delimiter/.style={xshift=-8pt}}}
\pgfkeys{tikz/mymatrix/.style={matrix of math nodes,nodes in empty cells,left delimiter={[},right delimiter={]},inner sep=1pt,outer sep=1.5pt,column sep=8pt,row sep=8pt,nodes={minimum width=20pt,minimum height=10pt,anchor=center,inner sep=0pt,outer sep=0pt}}}
\pgfkeys{tikz/mymatrixbrace/.style={decorate,thick}}

\newcommand*\mymatrixbraceright[4][m]{
    \draw[mymatrixbrace] (#1.west|-#1-#3-1.south west) -- node[left=2pt] {#4} (#1.west|-#1-#2-1.north west);
}

\newcommand*\mymatrixbracetop[4][m]{
    \draw[mymatrixbrace] (#1.north-|#1-1-#2.north west) -- node[above=2pt] {#4} (#1.north-|#1-1-#3.north east);
}

\tikzset{style green/.style={
    set fill color=green!50!lime!60,draw opacity=0.4,
    set border color=green!50!lime!60,fill opacity=0.1,
  },
  style cyan/.style={
    set fill color=cyan!90!blue!60, draw opacity=0.4,
    set border color=blue!70!cyan!30,fill opacity=0.1,
  },
  style orange/.style={
    set fill color=orange!90, draw opacity=0.8,
    set border color=orange!90, fill opacity=0.3,
  },
  style brown/.style={
    set fill color=brown!70!orange!40, draw opacity=0.4,
    set border color=brown, fill opacity=0.3,
  },
  style purple/.style={
    set fill color=violet!90!pink!20, draw opacity=0.5,
    set border color=violet, fill opacity=0.3,    
  },
  kwad/.style={
    above left offset={-0.1,0.23},
    below right offset={0.10,-0.36},
    #1
  },
  pion/.style={
    above left offset={-0.07,0.2},
    below right offset={0.07,-0.32},
    #1
  },
  poz/.style={
    above left offset={-0.03,0.18},
    below right offset={0.03,-0.3},
    #1
  },set fill color/.code={\pgfkeysalso{fill=#1}},
  set border color/.style={draw=#1}
}

\usepackage[normalem]{ulem}

\title{Symmetry reduction for testing $k$-block-positivity via extendibility}
\author{Qian Chen$^1$, Beno\^it Collins$^2$, Omar Fawzi$^1$ \\[2mm]
    {\small $^1$Universit\'e de Lyon, Inria, ENS de Lyon, UCBL, LIP, France} \\
    {\small $^2$ Kyoto University, Mathematics department
    }\\
  {\small chenqian.phys@gmail.com,
  collins@math.kyoto-u.ac.jp,
  omar.fawzi@ens-lyon.fr }}
\date{}

\begin{document}

\maketitle

\begin{abstract}
We study the problem of testing $k$-block-positivity via symmetric $N$-extendibility by taking the tensor product with a $k$-dimensional maximally entangled state. We exploit the unitary symmetry of the maximally entangled state to reduce the size of the corresponding semidefinite programs (SDP). For example, for $k=2$, the SDP is reduced from one block of size $2^{N+1} d^{N+1}$ to $\lfloor \frac{N+1}{2} \rfloor$
blocks of size $\approx O( (N-1)^{-1} 2^{N+1} d^{N+1} )$.
\end{abstract}

\section{Introduction}
A bipartite Hermitian operator $X \in \mathrm{Herm}_{d^2}(\mathbb{C})$ is said to be $k$-block-positive if $\tr(X \rho) \geq 0$ for any $\rho \in \mathrm{Sep}_{k}$, where $\mathrm{Sep}_k$ is the convex hull of states having Schmidt number at most $k$. A $k$-block-positive operator can act as a witness of having Schmidt rank larger than $k$~\cite{Stormer2013PLMOA}. It also closely connects to the notions of bound entanglement and distillability of entanglement, e.g., the 2-copy distillability conjecture \cite{PRXQuantum.3.010101} that asks whether $(\id+\alpha d \Pi_{d})^{\otimes 2}$ is nonnegative for all Schmidt rank-2 states. A bipartite Hermitian operator is said to be $k$-block-positive if its Hilbert-Schmidt product with any Schmidt number $k$ state is nonnegative. The set of $k$ block positive operators is the dual set of Schmidt number $k$ (or $k$-separable) states \cite{PhysRevA.61.040301,PhysRevA.63.050301,johnston2010NQitI,johnston2010NQitII}, and equivalent to $k$-positivity through Choi-Jamio{\l}kowski isomorphism \cite{skowronek2009cones}. 

\medskip

In this paper, we study testing $k$-block-positivity based on semidefinite programming (SDP).
To be more explicit, let $X$ be any bipartite Hermitian operators $\mathrm{Herm}_{d^2}(\mathbb{C})$ to be tested where $d$ is the local dimension, and we would like to find a lower bound on $h_{\mathrm{Sep}_k}(X) = \min_{\rho \in \mathrm{Sep}_k} \tr(X \rho)$. In order to achieve this, we consider extending $\mathbb{C}^{d} \to \mathbb{C}^{k} \otimes \mathbb{C}^{d}$ and then apply the trick of tensoring a $k$-dimensional maximally entangled projection to reduce the problem to a $1$-block-positivity problem \cite{johnston2012norms}.
This introduces an auxiliary system with dimension $k$, and converts $k$-block-positivity testing into block-positivity testing. We then use the standard SDP relaxation based on symmetric extensions of order $N$~\cite{caves2002DeFinetti,doherty2004DPShierarchy,christandl2007DeFinetti}. This gives rise to a semidefinite program whose optimal value gives a lower bound $\mathsf{SDP}_{k,N}(X)$ on $h_{\mathrm{Sep}_k}(X)$ (see Definition~\ref{definition:SDP-k-block-positivity-N-BSE} and Section~\ref{sec:SDP:ext-hierarchy} for details).
This SDP relaxation has multiple symmetries, in particular the unitary group in dimension $k$
acts as working as $U \otimes \id$, and the symmetric group of order $N$ whose implementation is defined as $\Delta_{B} : S_N \to \mathrm{U}((\mathbb{C}^{k} \otimes \mathbb{C}^{d})^{\otimes N})$ with $\Delta_{B}(\pi) \mapsto U_{\pi} \otimes U_{\pi}$.

The present paper studies the SDP reductions that arise from the symmetries, and estimates the computational resource that SDP may require. In general, the number of real variables required to parameterize an SDP underlying the set of $D \times D$ Hermitian positive definite matrices corresponds to the dimension of the space of $D \times D$ Hermitian matrices, which is $D^2$.
Without symmetry reduction, the size of a positive semidefinite matrix in $\mathsf{SDP}_{k,N}$ is $D=k^{N+1} \times d^{N+1}$.

After symmetry reduction, the positive semidefinite matrix $\rho$ is decomposed into blocks following the Schur-Weyl duality. Each block is associated with a Young diagram $\lambda$. For a Young diagram $\lambda$, we call the corresponding block the $\lambda$-block (see Eq.\eqref{eq:DiagramBlocks} and below explanation for details).
The action of permutations on these diagram-blocks are closed, hence the decomposition offers $\mathsf{SDP}_{k,N}$ a series of independent computations based on SDPs associated with irreducible representations of symmetric group.
That means, the permutational symmetry can be used independently in each block for the purpose of further reduction the diagram-blocks. In this paper, we focus on symmetry reduction stemming from unitary symmetries in the auxiliary spaces. The reduction stemming from permutational symmetry will be analyzed in subsequent work.

\medskip

The main result of this paper is presented below.
\begin{theorem}[$k$-block-positivity SDP symmetry reduction]
Denote $X_{(N)}=X \otimes \id_{d}^{\otimes (N-1)}$. We can write $\mathsf{SDP}_{k,N}(X)=\mathsf{SDP}_{k,N}^{\sym}(X)$ where $\mathsf{SDP}_{k,N}^{\sym}(X)$ is defined as follows:
\begin{align}
\mathsf{SDP}_{k,N}^{\sym}(X)
&:=
\min_{\{\rho_{\lambda} \in \mathrm{Pos}(\mathbb{C}^{d_{\lambda}} \otimes (\mathbb{C}^{d})^{\otimes (N+1)}), \lambda \vdash_{k} (N+k-1)\}} 
\tr [ ( \mathbb{P}_{\mathbb{Y}_{\lambda / (1^{k})}} \otimes X_{(N)} ) \rho_{\lambda} ], \\
\text{subject to } \:
&
\Delta_{\lambda}(\tau) \rho_{\lambda} = \rho_{\lambda}
, \
\forall \tau \in \mr{Cox}_{N}
, \ \text{and }
\tr\rho_{\lambda}=1.
\nonumber
\end{align}
Here,
\begin{itemize}
\item
$\lambda \vdash_k (N+k-1)$ denotes a Young diagram with $N+k-1$ boxes and exactly $k$ rows.
Denote $\mathrm{SYT}_{\lambda / (1^{k})}$ the set of standard Young tableaux based on $\lambda / (1^{k})$ associating projector $\mathbb{P}_{\mathbb{Y}_{\lambda / (1^{k})}}$ which is obtained by embedding $\mathbb{P}_{\mathbb{Y}_{\lambda / (1^{k})}}=\id_{\mathbb{Y}_{\lambda / (1^{k})}} \oplus 0_{\mathbb{Y}_{\lambda / (2,1^{k-2})}}$. Likewise, $\mathrm{SYT}_{\lambda / (2,1^{k-2})}$ denotes the set of standard Young tableaux based on skew shape $\lambda / (2,1^{k-2})$;
\item
Denote $\mathrm{Pos}(V)$ the set of positive definite matrices with respect to vector space $V$, and denote $d_{\lambda}$ the size of $\lambda$-block which should be given by $d_{\lambda}=\dim \mathbb{Y}_{\lambda / (1^{k})}+f^{\lambda / (2,1^{k-1})}$ as explained in Eq.\eqref{eq:formulas-DiagramBlock};
The block size of $\lambda$-block is $O(k^{N+1} (N-1)^{-\frac{k^2+k-2}{4}})$.
\item
$\mr{Cox}_{N}=\{ (j,j+1) \in S_{N} : k \leq j \leq N+k-2 \}$ is the set of Coxeter generators of $S_{N}$;
\item %(ii)
$\Delta_{\lambda} : S_N \to \mathrm{U}(\mathbb{C}^{d_{\lambda}}
 \otimes (\mathbb{C}^{d})^{\otimes (N+1)})$ is arisen from
$\Delta_{B} : S_N \to \mathrm{U}((\mathbb{C}^{k}
 \otimes \mathbb{C}^{d})^{\otimes N})$ with $\id_{A} \otimes \Delta_{B}(\pi) \mapsto U_{\pi}^{\lambda} \otimes U_{\pi}$ where $U^{\lambda}_{\pi}$ is the restricted representation to $S_{N}$.
\item There are at most $(N-1) d_{\lambda}^2 \times d^{N+1}$ many of constraints.
\end{itemize}
\end{theorem}
We illustrate the statement by looking at a simple example for testing $2$-positivity (i.e. $k=2$) and $X = \id_{d} \otimes \id_{d} + \alpha d \ketbra{\phi_{d}}{\phi_{d}}$ with parameter $\alpha$. We consider levels $N=1,2,3$ of the hierarchy. The sizes of the corresponding SDPs before and after symmetry reduction are listed in Tab~\ref{fig:examplePlots}. For $N=2$, the only Young diagram is $\ytableausetup{boxsize=0.5em} \begin{ytableau} {\scriptstyle \ } & {\scriptstyle \ } \\ {\scriptstyle \ } \end{ytableau}$ and for $N=3$ the Young diagrams are $\begin{ytableau} {\scriptstyle \ } & {\scriptstyle \ } \\ {\scriptstyle \ } & {\scriptstyle \ } \end{ytableau}$, $\begin{ytableau} {\scriptstyle \ } & {\scriptstyle \ } & {\scriptstyle \ } \\ {\scriptstyle \ } \end{ytableau}$; their $\lambda / (1^{k})$ are $\begin{ytableau} {\scriptstyle \bullet } & {\scriptstyle \ } \\ {\scriptstyle \bullet } \end{ytableau}$, $\begin{ytableau} {\scriptstyle \bullet } & {\scriptstyle \ } \\ {\scriptstyle \bullet } & {\scriptstyle \ } \end{ytableau}$, $\begin{ytableau} {\scriptstyle \bullet } & {\scriptstyle \ } & {\scriptstyle \ } \\ {\scriptstyle \bullet } \end{ytableau}$, respectively; their $\lambda / (2,1^{k-2})$ are $\begin{ytableau} {\scriptstyle \bullet } & {\scriptstyle \bullet } \\ {\scriptstyle \ } \end{ytableau}$, $\begin{ytableau} {\scriptstyle \bullet } & {\scriptstyle \bullet } \\ {\scriptstyle \ } & {\scriptstyle \ } \end{ytableau}$, $\begin{ytableau} {\scriptstyle \bullet } & {\scriptstyle \bullet } & {\scriptstyle \ } \\ {\scriptstyle \ } \end{ytableau}$ 
where $\bullet$ denotes the boxes that are not to be filled with numbers for having a standard Young tableaux.
For $N = 2$, there is only one Coxeter generator $\tau_{2}=(2,3)$ that permutes the second and third systems which are the two systems belonging to Bob. The corresponding $\Delta_{\lambda}(\tau)$ is
\begin{align*}
\Delta_{\begin{ytableau} {\scriptstyle \ } & {\scriptstyle \ } \\ {\scriptstyle \ } \end{ytableau}}((2,3))
=
\begin{pmatrix}
\frac{1}{2} & \frac{\sqrt{3}}{2} \\
\frac{\sqrt{3}}{2} & -\frac{1}{2}
\end{pmatrix}
\otimes
(2,3).
\end{align*}
On left side $\begin{pmatrix}
\frac{1}{2} & \frac{\sqrt{3}}{2} \\
\frac{\sqrt{3}}{2} & -\frac{1}{2}
\end{pmatrix}$ is the representation matrix of $\tau_{2}$ under $\begin{ytableau} {\scriptstyle \ } & {\scriptstyle \ } \\ {\scriptstyle \ } \end{ytableau}$; on right side $(2,3)$ stands for the natural representation of $\tau_{2}$. Similarly, for $N=3$, there are two Coxeter generators permuting Bob's systems: $\tau_{2}=(2,3)$ and $\tau_{3}=(3,4)$. The corresponding $\Delta_{\lambda}(\tau)$ are
\begin{align*}
&
\Delta_{\begin{ytableau} {\scriptstyle \ } & {\scriptstyle \ } \\ {\scriptstyle \ } & {\scriptstyle \ } \end{ytableau}}((2,3))
=
\begin{pmatrix}
\frac{1}{2} & \frac{\sqrt{3}}{2} \\
\frac{\sqrt{3}}{2} & -\frac{1}{2}
\end{pmatrix}
\otimes
(2,3)
, \qquad
\Delta_{\begin{ytableau} {\scriptstyle \ } & {\scriptstyle \ } \\ {\scriptstyle \ } & {\scriptstyle \ } \end{ytableau}}((3,4))
=
\begin{pmatrix}
-1 & 0 \\
0 & 1
\end{pmatrix}
\otimes
(3,4)
, \\
&
\Delta_{\begin{ytableau} {\scriptstyle \ } & {\scriptstyle \ } & {\scriptstyle \ } \\ {\scriptstyle \ } \end{ytableau}} ((2,3))
=
\begin{pmatrix}
\frac{1}{2} & \frac{\sqrt{3}}{2} & 0 \\
\frac{\sqrt{3}}{2} & -\frac{1}{2} & 0 \\
0 & 0 & 1
\end{pmatrix}
\otimes (2,3),
, \qquad
\Delta_{\begin{ytableau} {\scriptstyle \ } & {\scriptstyle \ } & {\scriptstyle \ } \\ {\scriptstyle \ } \end{ytableau}} ((3,4))
=
\begin{pmatrix}
1 & 0 & 0 \\
0 & \frac{1}{3} & \frac{2\sqrt{2}}{3} \\
0 & \frac{2\sqrt{2}}{3} & -\frac{1}{3}
\end{pmatrix}
\otimes (3,4),
\end{align*}
The corresponding $\Delta_{\lambda}$ is the representation defined by irreducible representation $\lambda$ tensoring canonical permutation representation. One could refer to Eq.\eqref{eq:ex-Delta2,1-2,3},\eqref{eq:ex-Delta2,2},\eqref{eq:ex-Delta3,1}.
The minimal values of hierarchies $N=1,2,3$ are plotted in Fig.~\ref{fig:examplePlots} which were done using Intel Core i5 with16 GB of RAM memory \cite{qian_chen_2025_kbp-example}. The reduced SDPs is solved faster than unreduced SDPs.
\begin{table}[h!]
	\centering
	\small
	\begin{tabular}{|c|c|c|c|c|c|c|c|c|}
		\hline
		$d$
		%\multirow{3}{*}{$k$}
		& \multicolumn{3}{c|}{$N=2$}
		& \multicolumn{5}{c|}{$N=3$}\\
		\cline{2-9}
		& unreduced
		& size of $\rho_{\ytableausetup{boxsize=0.35em} \begin{ytableau} {\scriptstyle \ } & {\scriptstyle \ } \\ {\scriptstyle \ } \end{ytableau}}$
		& $d_{\begin{ytableau} {\scriptstyle \ } & {\scriptstyle \ } \\ {\scriptstyle \ } \end{ytableau}}$
		& unreduced
		& size of $\rho_{\begin{ytableau} {\scriptstyle \ } & {\scriptstyle \ } \\ {\scriptstyle \ } & {\scriptstyle \ } \end{ytableau}}$
		& size of $\rho_{\begin{ytableau} {\scriptstyle \ } & {\scriptstyle \ } & {\scriptstyle \ } \\ {\scriptstyle \ } \end{ytableau}}$
		& $d_{\begin{ytableau} {\scriptstyle \ } & {\scriptstyle \ } \\ {\scriptstyle \ } & {\scriptstyle \ } \end{ytableau}}$
		& $d_{\begin{ytableau} {\scriptstyle \ } & {\scriptstyle \ } & {\scriptstyle \ } \\ {\scriptstyle \ } \end{ytableau}}$ \\
		\hline
		$2$ & 64 & $16=2^3 \times 2$ & 2 & 256 & $32=2^4 \times 2$ & $48=2^4 \times 3$ & 2 & 3 \\
		\hline
		$3$& 216& $54=3^3 \times 2$ & 2 & 1296 & $162=3^4 \times 2$ & $243=3^4 \times 3$ & 2 & 3 \\
		\hline
		$4$& 512 & $128=4^3 \times 2$ & 2 & 4096 & $512=4^4 \times 2$ & $768=4^4 \times 3$ & 2 & 3 \\
		\hline
		$5$& 1000 & $250=5^3 \times 2$ & 2 & 10000 & $1250=5^4 \times 2$ & $1875=5^4 \times 3$ & 2 & 3 \\
		\hline
	\end{tabular}
	\caption{The comparison of the reductions obtained by considering unitary invariance under the action of $\mathrm{U}(k)^{\otimes (N+k-1)}$ on the auxiliary spaces. Note that the size of $\rho_{\lambda}$ is $d^{N+1} \cdot d_{\lambda}$.}
\end{table}
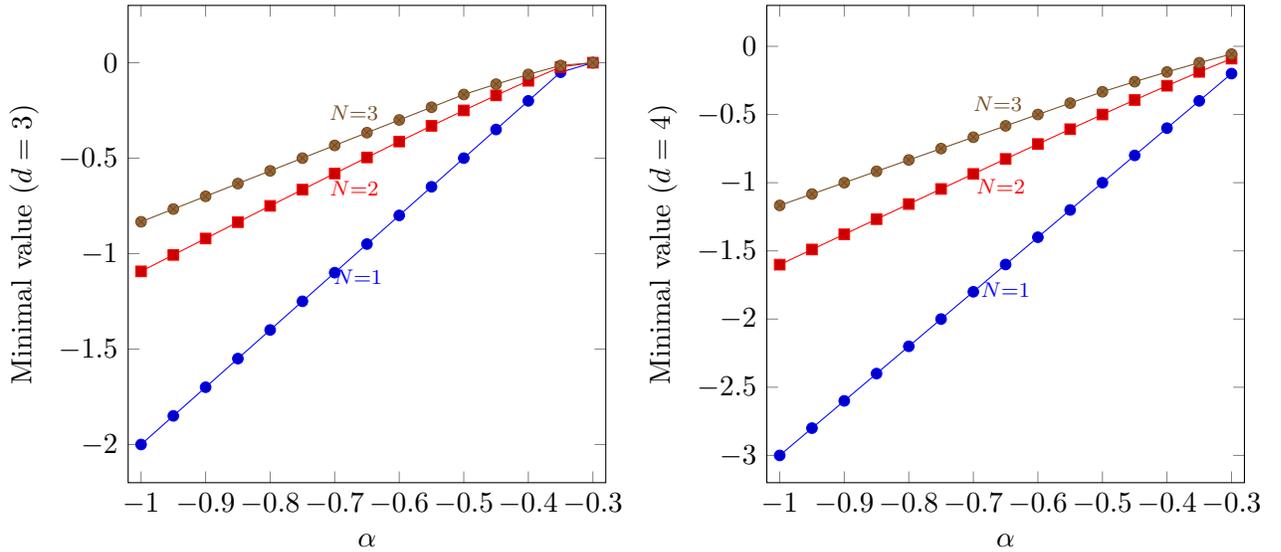
\begin{figure}[ht]
\begin{tikzpicture}
\begin{axis}[
width=2.5in,
height=2.5in,
at={(1.048in,0.726in)},
scale only axis,
xmin=-1.02,
xmax=-0.28,
ymin=-2.2,
ymax=0.3,
xlabel=$\alpha$,
ylabel=Minimal value ({$d=3$}),
axis background/.style={fill=white},
legend style={legend cell align=left, align=left, draw=white!15!black}]
\addplot table[row sep=crcr]{
-1.0 -2 \\
-0.95 -1.85 \\
-0.9 -1.7 \\
-0.85 -1.55 \\
-0.8 -1.4 \\
-0.75 -1.25 \\
-0.7 -1.1 \\
-0.65 -0.95 \\
-0.6 -0.8 \\
-0.55 -0.65 \\
-0.5 -0.5 \\
-0.45 -0.35 \\
-0.4 -0.2 \\
-0.35 -0.05 \\
-0.3 -2.6491e-16 \\
}
node[below=1mm,pos=0.5] {${\scriptstyle{N=1}}$};
\addplot table[row sep=crcr]{
-1.0 -1.09307 \\
-0.95 -1.00697 \\
-0.9 -0.921066 \\
-0.85 -0.835391 \\
-0.8 -0.75 \\
-0.75 -0.664963 \\
-0.7 -0.580374 \\
-0.65 -0.496361 \\
-0.6 -0.413104 \\
-0.55 -0.330859 \\
-0.5 -0.25 \\
-0.45 -0.171084 \\
-0.4 -0.0949 \\
-0.35 -0.0228 \\
-0.3 -2.7195e-16 \\
}
node[below=1mm,pos=0.5] {${\scriptstyle{N=2}}$};
\addplot table[row sep=crcr]{
-1.0 -0.833333 \\
-0.95 -0.766667 \\
-0.9 -0.7 \\
-0.85 -0.633333 \\
-0.8 -0.566667 \\
-0.75 -0.5 \\
-0.7 -0.433333 \\
-0.65 -0.366667 \\
-0.6 -0.3 \\
-0.55 -0.233333 \\
-0.5 -0.166667 \\
-0.45 -0.1125 \\
-0.4 -0.0615 \\
-0.35 -0.0146 \\
-0.3  -2.8998e-16 \\
}
node[above=1mm,pos=0.5] {${\scriptstyle{N=3}}$};
\end{axis}
\end{tikzpicture}
\begin{tikzpicture}
\begin{axis}[
width=2.5in,
height=2.5in,
at={(1.048in,0.726in)},
scale only axis,
xmin=-1.02,
xmax=-0.28,
ymin=-3.2,
ymax=0.3,
xlabel=$\alpha$,
ylabel=Minimal value ({$d=4$}),
axis background/.style={fill=white},
legend style={legend cell align=left, align=left, draw=white!15!black}]
\addplot table[row sep=crcr]{
-1.0 -3 \\
-0.95 -2.8 \\
-0.9 -2.6 \\
-0.85 -2.4 \\
-0.8 -2.2 \\
-0.75 -2 \\
-0.7 -1.8 \\
-0.65 -1.6 \\
-0.6 -1.4 \\
-0.55 -1.2 \\
-0.5 -1 \\
-0.45 -0.8 \\
-0.4 -0.6 \\
-0.35 -0.4 \\ 
-0.3 -0.2 \\ 
}
node[below=1mm,pos=0.5] {${\scriptstyle{N=1}}$};
\addplot table[row sep=crcr]{
-1.0 -1.60128 \\
-0.95 -1.48988 \\
-0.9 -1.37862 \\
-0.85 -1.26752 \\
-0.8 -1.15664 \\
-0.75 -1.04601 \\
-0.7 -0.935695 \\
-0.65 -0.825789 \\
-0.6 -0.71641 \\
-0.55 -0.60773 \\
-0.5 -0.5 \\
-0.45 -0.393599 \\
-0.4 -0.289117 \\
-0.35 -0.1875 \\
-0.3 -0.0902969 \\
}
node[below=1mm,pos=0.5] {${\scriptstyle{N=2}}$};
\addplot table[row sep=crcr]{
-1.0 -1.16667 \\
-0.95 -1.08333 \\
-0.9 -1 \\
-0.85 -0.916667 \\
-0.8 -0.833333 \\
-0.75 -0.75 \\
-0.7 -0.666667 \\
-0.65 -0.583333 \\
-0.6 -0.5 \\
-0.55 -0.416667 \\
-0.5 -0.333333 \\
-0.45 -0.2594 \\
-0.4 -0.1880 \\
-0.35 -0.1201 \\
-0.3 -0.0569 \\
}
node[above=1mm,pos=0.5] {${\scriptstyle{N=3}}$};
\end{axis}
\end{tikzpicture}
  \caption{The minimal values for varying $\alpha$.
  }
  \label{fig:examplePlots}
\end{figure}

The paper is organized as follows. In Section~\ref{sec:KPB:SDP-hierarchy} we present our notation and terminology and introduce the reader to the Schmidt number of density operators, $k$-block-positivity, the trick of $k$-extension that tensors $k$-dimensional maximally entangled projector, and the extendibility hierarchy. In Section~\ref{sec:symmetric-reduction-unitary} we implement unitary twirling for symmetry reduction. Using dualization, we convert $\bar{U} \otimes U$ symmetry, which arises from conjugation action on the $k$-dimensional maximally entangled projector, to $U^{\otimes k}$ symmetry.
We then apply Schur-Weyl duality to block diagonalize the tensor space, leading to Theorem~\ref{thm:Twirling-PartialUnitary}. Section~\ref{sec:PermutationalSym} follows the block structure, showing how to implement permutational symmetry in Subsection~\ref{subsec:symmetric-reduction-permutation}, and analyzing the asymptotic ratio of sizes $\dim \mathbb{Y}_{\lambda / (1^{k})}$ (contributing to objective function) and $\dim \mathbb{Y}_{\lambda / (2,1^{k-2})}$ (balancing the trace due to permutation constraints) in Subsection~\ref{subsec:Asymptotic Diagram-blockSize}.

%%%%%%%%

%------------------------------------------------

%%%%%%%%
 
\section{Semidefinite programming relaxations for $k$-block-positivity} \label{sec:KPB:SDP-hierarchy}

\paragraph{Notation.} Let the symbol $\mathbb{M}$ stand for matrix spaces and $\id$ for the identity operator. By default, we set $\mathbb{C}^{d}$ as the unextended spaces for Alice and Bob. Denote the Schmidt rank of a pure bipartite state $\ket{v}$ by $sr(v)$. Denote the Schmidt number of a mixed bipartite state $\rho$ by $sn(\rho)$.

We define $k$-extension by introducing auxiliary $\mathbb{C}^{k}$ on each subsystem $\mathbb{C}^{d}$ via $\mathbb{C}^{d} \to \mathbb{C}^{k} \otimes \mathbb{C}^{d}$.
The local subsystem after $k$-extension $\mathcal{H}_{A} \cong \mathcal{H}_{B} \cong \mathbb{C}^{k} \otimes \mathbb{C}^{d} \cong \mathbb{C}^{kd}$.
In the later sections, we will introduce dualization $\mathrm{Alt}^{k-1} \mathbb{C}^{k} \cong \mathbb{C}^{k}$ on Alice's auxiliary and still denote $\mathcal{H}_{A} = (\mathrm{Alt}^{k-1} \mathbb{C}^{k}) \otimes \mathbb{C}^{d}$.

Denote the normalized projection of $k$-dimensional maximally entangled state by $\ket{\phi_{k}}$.
We simplify $\id_{A} \otimes \pi$ to $\pi$ when there is no confusion, where $\pi$ is a permutation on Bob's extension.

The symbol $\lambda \vdash_{k} n$ means $\lambda$ a Young diagram with $n$ boxes and exactly $k$ rows. The symbol $\mathbb{Y}_{\lambda}$ stands for the Specht module associative to the Young diagram $\lambda$. Symbols $\mathbb{U}_{k, \lambda}$ and $\mathbb{U}_{d, \lambda}$ for irreducible representations of unitary groups $\mathrm{U}(k)$ and $\mathrm{U}(d)$ respectively.
Denote the skew Young diagram by $\lambda / \mu$ where $\lambda$ and $\mu$ are two Young diagrams with $\lambda  \supset \mu$.

We denote Schur basis under $\lambda$ by $\ket{p_{\lambda}} \otimes \ket{q_{\lambda}}$ or $\ket{p_{\lambda}, q_{\lambda}}$. The letter $T$ denotes the Schur transform that sends the computational basis to the Schur basis, the calligraphic letter $\mathcal{T}$ denotes the twirling operation, and $\mathcal{T}_{U}$ for auxiliary $\mr{U}(k)$-twirling \cite{Bartlett:2006tzx}, in particular.

%%%%%%%%

%------------------------------------------------

%%%%%%%%

\subsection{$k$-block positivity, $k$-extension, and the related semidefinite programming}

We present the mathematical setup for the $k$-block-positivity. Consider a bipartite system $\mathbb{C}^{d} \otimes \mathbb{C}^{d}$. Any bipartite pure state with at most Schmidt rank $k$, can be written into the form below:
\begin{align}
\ket{\psi}=\sum_{p=1}^{k} \ket{z_{p}} \otimes \ket{w_{p}}, \: \text{ where for all $p$, both } \ket{z_{p}}, \: \ket{w_{p}} \in \mathbb{C}^{d}.
\end{align}
The pure states with at most Schmidt rank $k$ form a subset of the set of all pure states,
\begin{align}
\mathrm{SR}_{k} (d) = \{ \ket{\psi} \in \mathbb{C}^{d} \otimes \mathbb{C}^{d} : sr({\psi}) \leq k \}. \label{eq:Set:SRk}
\end{align}
A Hermitian operator $X \in \mathrm{Herm}_{d^2 \times d^2}(\mathbb{C})$ is said to be $k$-block-positive if $X$'s expectation value is nonnegative for all the members of $\mathrm{SR}_{k} (d)$, i.e., $\bra{\psi} X \ket{\psi} \geq 0$ for all $\ket{\psi} \in \mathrm{SR}_{k} (d)$.
A mixed state $\rho$ is said to have the Schmidt number $k$, denoted by $sn(\rho)=k$, if there exists an ensemble $\{ p_{i} , \psi_{i} \}$ such that $\rho=\sum_{i} p_{i} \ketbra{\psi_{i}}{\psi_{i}}$ and all $sr(\psi_{i}) \leq k$ \cite{PhysRevA.61.040301,PhysRevA.63.050301}.
The set of Schmidt number $k$ states is denoted by
\begin{align}
\mathrm{SN}_{k} (d) = \{ \rho \in \mathrm{Herm}(\mathbb{C}^{d} \otimes \mathbb{C}^{d})_{+} : sn(\rho) \leq k \}. \label{eq:Set:SNk}
\end{align}
A Hermitian operator $X \in \mathrm{Herm}_{d^2 \times d^2}(\mathbb{C})$ is said to be $k$-block-positive if and only if $\tr (X \rho) \geq 0$ for all $sn(\rho) \leq k$. The following optimization problem is formulated to test the $k$-block-positivity.
\begin{definition}[Optimization: $k$-block-positivity] \label{definition:SDP-k-block-positivity-original}
A Hermitian operator $X \in \mathrm{Herm}_{d^2 \times d^2}(\mathbb{C})$ is $k$-block-positive if and only if the following optimization problem gives nonnegative optimal value,
\begin{align}
&
\min \tr{X} \rho,
\\
&
\text{subject to } \:
\rho \in \mathrm{SN}_{k} (d), \: \text{ and } \tr\rho=1.
\nonumber
\end{align}
\end{definition}
Since $\mathrm{SN}_{1} \subset \mathrm{SN}_{2} \subset \cdots \subset \mathrm{SN}_{k} \subset \cdots \subset \mathrm{SN}_{d-1} \subset \mathrm{SN}_{d}$, the minimal values satisfy the sequence of inequalities:
\begin{align}
\min_{\rho \in \mathrm{SN}_{d}} \tr{X} \rho \leq \min_{\rho \in \mathrm{SN}_{d-1}} \tr{X} \rho \leq \cdots \leq \min_{\rho \in \mathrm{SN}_{k}} \tr{X} \rho \leq \cdots \leq \min_{\rho \in \mathrm{SN}_{2}} \tr{X} \rho \leq \min_{\rho \in \mathrm{SN}_{1}} \tr{X} \rho.
\end{align}
\begin{definition}[$k$-extension]
We define $k$-extension $\mathbb{C}^{d} \to \mathbb{C}^{k} \otimes \mathbb{C}^{d}$.
For any $X \in \mathrm{Herm}(\mathbb{C}^{d} \otimes \mathbb{C}^{d})$, its $k$-extension $X_{k} \in \mathrm{Herm}(\mathbb{C}^{kd} \otimes \mathbb{C}^{kd})$ is defined as,
\begin{align}
&
X_{k}:=
\ketbra{\phi_{k}}{\phi_{k}} \otimes X, \:
\text{ where }
\ket{\phi_{k}}=\sum_{i=1}^{k} \frac{1}{\sqrt{k}}\ket{i^{*}i},
\end{align}
On the other hand, any $\rho_{k} \in \mathrm{Herm}(\mathbb{C}^{kd} \otimes \mathbb{C}^{kd})_{+}$ can be written as
\begin{align}
&
\rho_{k}:=
\sum_{i_{0},i_{1},j_{0},j_{1}=1}^{k}
\ketbra{i_{0} i_{1} }{j_{0} j_{1}}
\otimes
\rho_{i_{0} i_{1}, j_{0} j_{1}}, \:
\text{ where }
\rho_{i_{0} i_{1}, j_{0} j_{1}}
\in \mathbb{M}_{d^2 \times d^2}(\mathbb{C}).
\end{align}
\end{definition}
\begin{lemma}[$k$-block-positivity testing via $k$-extension] \label{lemma:SDP-k-block-positivity}
$X$ is $k$-block positive if and only if $X_{k}$ is block positive. Thus, the $k$-block-positivity testing can be formulated as
\begin{align}
&
\min \tr{ X_{k} \rho_{k}},
\\
&
\text{subject to } \:
\rho_{k} \in \mathrm{Sep}(\mathbb{C}^{kd} \otimes \mathbb{C}^{kd}), \: \text{ and } \tr\rho_{k}=1.
\nonumber
\end{align}
\end{lemma}
Lemma~\ref{lemma:SDP-k-block-positivity} can be proved by considering the parameterization $\ket{\psi}=\sum_{p=1}^{k} \ket{z_{p} \otimes w_{p}}$ for Schmidt rank $k$ pure state, then one has Hermitian polynomial
\begin{align}
\bra{\psi} X \ket{\psi}
=
\sum_{p,q=1}^{k} \sum_{i,j,l,m=1}^{d} X_{lm,ij} z_{ip} w_{jp} \bar{z}_{lq} \bar{w}_{mq}
=
\bra{ z \otimes w }
(k \ketbra{\phi_{k}}{\phi_{k}} \otimes X)
\ket{z \otimes w},
\end{align}
with $\ket{z}=\sum_{p=1}^{k} \ket{ p \otimes z_{p} }$ and $\ket{w}=\sum_{q=1}^{k} \ket{ q \otimes w_{q} }$.
Through a straightforward calculation, one may realize that $k$-extension amounts to purifying Schmidt number $k$ to Schmidt number $1$.
The polynomial $X(z,w) \equiv \bra{\psi} X \ket{\psi}$ could be viewed as a generalization of \cite{harrow2019limitations}.
Although $\| \ket{\psi} \|_{2}$ is unnecessarily equal to $\| \ket{z \otimes w}\|_{2}$,
positivity testing only cares about the sign of the minimal value, and therefore setting $\tr\rho_{k}=c>0$ with $c \neq 1$ is allowed.
\begin{remark}
Even though in general the minimal values given from the optimization problems Definition~\ref{definition:SDP-k-block-positivity-original} and Lemma~\ref{lemma:SDP-k-block-positivity} are different, their signs are the same.
\end{remark}

%%%%%%%%

%------------------------------------------------

%%%%%%%%

\subsection{SDP relaxation via extendibility hierarchy} \label{sec:SDP:ext-hierarchy}

The $k$ block-positivity test is now converted into block-positivity testing through the lemma~\ref{lemma:SDP-k-block-positivity}, the optimization problem could be then solved by introducing relaxation such as extendibility hierarchy and Doherty-Parrilo-Spedalieri (DPS) hierarchy \cite{doherty2004DPShierarchy,christandl2007DeFinetti,PhysRevA.80.052306,fang2021SOS}. In this paper, we will use the extendibility hierarchy.

Denote the symmetric extension of the $N$-level by $\rho_{k} \mapsto \rho_{k,N}$ where $N$ is the number of Bob's copies, and correspondingly we extend $X_{k} \mapsto X_{k,N}$ by $X_{k,N}=\Pi_{k} \otimes \id_{k}^{\otimes (N-1)} \otimes X \otimes \id_{d}^{\otimes (N-1)}$.

A bipartite state $\rho_{AB}$ is said to be $N$-(symmetric) extendible \cite{christandl2007DeFinetti}, if Bob's (likewise for Alice's) system can be extended into $N$-partite $\rho_{AB_{1} \cdots B_{N}}$, such that the Bob's extension is $N$-exchangeable $\rho_{AB_{1} \cdots B_{N}}=(\id_{A} \otimes \pi_{B}) \rho_{AB_{1} \cdots B_{N}} (\id_{A} \otimes \pi^{-1}_{B})$ or $N$-Bose-exchangeable $\rho_{AB_{1} \cdots B_{N}}=(\id_{A} \otimes \pi_{B}) \rho_{AB_{1} \cdots B_{N}}=\rho_{AB_{1} \cdots B_{N}} (\id_{A} \otimes \pi_{B})$ for all permutation $\pi_{B} \in S_{N}$, meanwhile $\rho_{AB} $ can be retrieved via partial trace the extension, i.e., $\rho_{AB} \equiv \rho_{AB_{1}}=\tr_{B_{2} \cdots B_{N}}\rho_{AB_{1} \cdots B_{N}}$.

In our problem, permutation is defined for the $(\mathbb{C}^{kd})^{\otimes N} \cong (\mathbb{C}^{k} \otimes \mathbb{C}^{d})^{\otimes N}$ due to $k$-extension $\otimes$, given by the following map:
\begin{align}
\Delta_{B} : S_{N} \to \mathrm{U}((\mathbb{C}^{k} \otimes \mathbb{C}^{d})^{\otimes N}).
\end{align}
A state is separable if and only if it is infinitely-exchangeable, or infinitely-Bose-exchangeable. The Bose exchangeability is stronger than the $N$ exchangeability, with faster convergence in quantum de Finetti theorem \cite{christandl2007DeFinetti}, but the limit case is the same. From now on, we set $\rho_{k,N}$ to be a $N$-symmetric bosonic extension ($N$-BSE) of $\rho_{k}$ provided that $\rho_{k}$ is $N$-Bose-exchangeable,
\begin{align}
\rho_{k,N}
=
( \id_{A} \otimes \Delta_{B}(\pi) ) \rho_{k,N}
=
\rho_{k,N} ( \id_{A} \otimes \Delta_{B}(\pi) ), \ \forall \pi \in S_{N}, \:
\text{ where } \rho_{k}=\tr_{\mathcal{H}_{B}^{\otimes (N-1)}}\rho_{k,N}.
\end{align}
We define SDP with $N$-BSE $\rho_{k,N}$ instead of $N$ -Bose-extendible $\rho_{k}$.
\begin{definition}[$k$-block-positivity testing SDP with N-BSE] \label{definition:SDP-k-block-positivity-N-BSE}
The $N$ level of extendibility hierarchy SDP is defined as below,
\begin{align}
&
\mathsf{SDP}_{k,N}(X):=
\min \tr{ X_{k,N} \rho_{k,N}}, \\
&
\text{subject to } \:
\rho_{k,N} \in \mathrm{Bos}_{\mathcal{H}_{B}}(\mathcal{H}_{A} , \mathcal{H}_{B}^{\otimes N}), \ \text{ and } \tr\rho_{(k, N)}=1.
\nonumber
\end{align}
\end{definition}
Since the SDP is valued by the points in permutation symmetric space, we can define projector $P_{k, N}=\frac{1}{N!} \sum_{\pi \in S_{N}} \id_{A} \otimes \Delta(\pi)$. By this, the SDP problem can be translated into a solving min-eigenvalue problem following Courant-Fischer-Weyl min-max theorem.
\begin{proposition} \label{pro:SDP-Eigenvalue}
Define $P_{k, N}=\frac{1}{N!} \sum_{\pi \in S_{N}} \id_{A} \otimes \Delta_{B}(\pi)$, the $\mathsf{SDP}_{k,N}(X)$ can be computed by solving the following minimal eigenvalue problem,
\begin{align}
\mathsf{SDP}_{k,N}(X)
=
\min \left\{ {\mathrm{eig}} \left[
\frac{1}{(N!)^2} \sum_{ \pi , \sigma \in S_{N} } (\id_{A} \otimes \Delta_{B}(\pi) ) X_{(k,N)} (\id_{A} \otimes \Delta_{B}(\sigma) )
\right] \right\}.
\end{align}
\end{proposition}
In many cases, computing the minimal eigenvalue is computationally simpler than solving the associated SDP. However, since this work also addresses the estimation of computational resources, the matrix dimension involved in the SDP serves as a metric for resource quantification. Consequently, our subsequent analysis will focus on the SDP framework.

\medskip

Solving the SDP Definition~\ref{definition:SDP-k-block-positivity-N-BSE} requires tremendous computational resource. To feel it, one may look at the size of the input PSD matrices, which is $(kd)^{N+1} \times (kd)^{N+1}$. In order to run the SDP more efficiently, we will make use of $\bar{U} \otimes U$-symmetry for symmetry reduction, then take $S_{N}$-symmetry as the constraints in SDP.

%%%%%%%%

%------------------------------------------------

%%%%%%%%

\section{The SDP reduction via $\mathrm{U}(k)$ symmetry} \label{sec:symmetric-reduction-unitary}

This section is dedicated to symmetry reduction based on the $\bar{U} \otimes U$-symmetry on auxiliary spaces. We will convert the $\bar{U} \otimes U$ symmetry to $U^{\otimes k}$, then use Schur transform to diagonalize the twirled state, labeling the blocks by Young diagrams from tensor decomposition.

%%%%%%%%

%------------------------------------------------

%%%%%%%%

\subsection{Schur transform}
\label{sec:schur}
Before symmetry reduction, let us briefly review the Schur transform. The tensor representation $V^{\otimes n}$ for any $n$ admits a decomposition due to Schur-Weyl duality,
\begin{align}
V^{\otimes n} \cong \bigoplus_{\lambda \vdash n} \mathbb{Y}_{\lambda} \otimes \mathbb{U}_{\lambda},
\end{align}
with $\mathbb{Y}_{\lambda}$ irreducible representation of $S_{n}$ and $\mathbb{U}_{\lambda}$ irreducible representation of $\mathrm{GL}(V)$.
Set $V=\mathbb{C}^{k}$. The isomorphism is realized by Schur transform \cite{bacon2005Schur},
\begin{align}
T: (\mathbb{C}^{k})^{\otimes n} \to \bigoplus_{\lambda \vdash_{k} n} \mathbb{Y}_{\lambda} \otimes \mathbb{U}_{\lambda}.
\end{align}
Write Schur basis as $\ket{ \lambda , p_{\lambda} , q_{\lambda} }=\ket{\lambda} \otimes \ket{p_{\lambda}} \otimes \ket{q_{\lambda} }$ with $p_{\lambda}=1, \cdots , \dim \mathbb{Y}_{\lambda}$ and $q_{\lambda}=1, \cdots , \dim \mathbb{U}_{\lambda}$.
The Schur transform $T$ sends the computational basis to the Schur basis, $\ket{ \vec{i} } \overset{T}{\to} \ket{ \lambda , p_{\lambda} , q_{\lambda} }$ where $\vec{i} \equiv i_{1} \cdots i_{n}$.
The labeling state $\ket{\lambda}$ might be omitted to keep the notation light.
We adopt the English notation for Young diagrams and tableaux. Let us further explain the definition of the Schur basis:
\begin{itemize}
\item $\lambda$ denotes a Young diagram, and $\lambda \vdash n$ means $\lambda$ having $n$ boxes. In the case of $\mathrm{U}(k)$, $\lambda$ is restricted to having at most $k$ rows, expressed by $\lambda \vdash_{k} n$. Denote by $h_{\lambda}(i,j)$ the hook length with respect to the box $(i,j)$.

\item The data $p_{\lambda}$ labels a standard Young tableau under $\lambda$ thus labels a basis vector for $\mathbb{Y}_{\lambda}$ i.e., $p_{\lambda}=1, \cdots , \dim \mathbb{Y}_{\lambda}$ where the dimension is given by the hook length formula
\begin{align}
\dim \mathbb{Y}_{\lambda}=\frac{n!}{ \prod_{(i,j) \in \lambda} h_{\lambda}(i,j) }.
\end{align}

\item The data $q_{\lambda}$ labels a semistandard Young tableau thus labels a basis vector for $\mathbb{U}_{\lambda}$, i.e., $q_{\lambda}=1, \cdots , \dim \mathbb{U}_{\lambda}$ where the dimension is given by
\begin{align}
\dim \mathbb{U}_{\lambda}=\prod_{(i,j) \in \lambda} \frac{k+j-i}{h_{\lambda}(i,j) }.
\label{eq:Dim-unitary-irrep}
\end{align}

\end{itemize}
The orthonormal basis for $\mathbb{Y}_{\lambda}$ and $\mathbb{U}_{\lambda}$ can be constructed systematically \cite{Goodman2009Symmetry,ceccherini2010RepsSym,Fulton2013RT}. In this paper, we take by default orthonormal bases, i.e. $\inner{p'_{\lambda'} , q'_{\lambda'}}{p_{\lambda} , q_{\lambda}}=\delta_{\lambda' \lambda} \delta_{p'_{\lambda}p_{\lambda}} \delta_{q'_{\lambda}q_{\lambda}}$.

\smallskip

According to Schur-Weyl duality, the tensor representation $U^{\otimes n}$ and permutation $\pi \in S_{n}$ are block-diagonal with respect to $\lambda$ under the Schur basis,
\begin{align}
U^{\otimes n}
& \cong \bigoplus_{\lambda \vdash n} \id_{\mathbb{Y}_{\lambda}} \otimes U_{\lambda}
\cong \bigoplus_{\lambda \vdash n} (\dim \mathbb{Y}_{\lambda})U_{\lambda},
\ \text{with} \
U^{\otimes n}=\bigoplus_{\lambda \vdash n} T^{-1} (\id_{\mathbb{Y}_{\lambda}} \otimes U_{\lambda}) T,
\\
\pi
& \cong \bigoplus_{\lambda \vdash n} \pi_{\lambda} \otimes \id_{\mathbb{U}_{\lambda}}
\cong \bigoplus_{\lambda \vdash n} (\dim \mathbb{U}_{\lambda})\pi_{\lambda},
\ \text{with} \
\pi=\bigoplus_{\lambda \vdash n} T^{-1} (\pi_{\lambda} \otimes \id_{\mathbb{U}_{\lambda}}) T.
\end{align}
Here, $\dim \mathbb{Y}_{\lambda}$ and $\dim \mathbb{U}_{\lambda}$ are also the respective multiplicities of $U_{\lambda}$ and $\pi_{\lambda}$.

%%%%%%%%

%------------------------------------------------

%%%%%%%%

\subsection{Dualization method: from $\bar{U} \otimes U$ symmetry to $U^{\otimes k}$}

The goal of this subsection is to convert the $\bar{U} \otimes U$ symmetry to $U^{\otimes k}$ symmetry to apply the Schur-Weyl and Schur transform to our situation.

The $\ket{\phi_{k}}$ is the maximal entangled state in the auxiliary $\mathbb{C}^{k} \otimes \mathbb{C}^{k}$, invariant under $\bar{U} \otimes U$. In order to apply Schur-Weyl duality, we convert the $\bar{U} \otimes U$ symmetry by the exterior product $\mathrm{Alt}^{k-1} \mathbb{C}^{k} \cong \mathbb{C}^{k}$. If $k=1$, we need to do nothing. When $k \geq 2$, Alice's basis can be equivalently described by below isomorphism with the dual basis $\{ \ket{i_{2} \cdots i_{k}} \}$ itself,
\begin{align}
\ket{i^{*}} \leftrightarrow \frac{1}{\sqrt{(k-1)!}} \sum_{i_{2} , \ldots , i_{k}=1}^{k} \epsilon_{i_{2} \ldots i_{k} i} \ket{i_{2} \cdots i_{k} }.
\label{eq:DualBasis}
\end{align}
One can show $\frac{1}{\sqrt{k!}} \sum_{i_{1} , \ldots , i_{k}=1}^{k} \epsilon_{i_{1} \ldots i_{k}} \ket{i_{1} \ldots i_{k}}$ indeed the stabiliser of $\frac{ U^{\otimes k} }{\det U}$. So we can build isomorphism between $\frac{1}{\sqrt{k!}} \sum_{i_{1} , \ldots , i_{k}=1}^{k} \epsilon_{i_{1} \ldots i_{k}} \ket{i_{1} \ldots i_{k}}$ and $\ket{\phi_{k}}$, then denote $\Pi_{k}$ the Young projector associated with the Young diagram $(1^{k})$, which is dual to the $1$-rank projector $\ketbra{\phi_{k}}{\phi_{k}}$,
\begin{align}
\Pi_{k}=\sum_{i,i'=1}^{k} \frac{\epsilon_{i_{1} \ldots i_{k}}\epsilon_{i'_{1} \ldots i'_{k}}}{k!} \ketbra{i_{1} \ldots i_{k}}{i'_{1} \ldots i'_{k}}
, \text{ satisfying }
\Pi_{k}=
U^{\otimes k} \Pi_{k} {U^{\dagger}}^{\otimes k} \ \forall U \in \mathrm{U}(k).
\end{align}

\smallskip We then move to the $N$-extension. Label the Alice's system by integers from $[1,k-1]$ and Bob's system by $[k,N+k-1]$, then
\begin{align}
&
\rho_{k,N}
=
\sum_{\vec{i} , \vec{j} \in [k]^{N+k-1} }\ketbra{ \vec{i} }{ \vec{j} }
\otimes
\rho_{\vec{i} , \vec{j}},
\ \text{where} \
\rho_{\vec{i} , \vec{j}}=
\sum_{i_{0},j_{0}=1}^{k}
\epsilon_{a_{2} \ldots a_{k} i_{0}} \epsilon_{'_{2} \ldots a'_{k} j_{0}}\rho_{i_{0} i_{1} \ldots i_{N} , j_{0} j_{1} \ldots j_{N}},
\label{eq:Blocks-Nextension-states-dual}
\\
&
X_{k,N}
=
\Pi_{k} \otimes \id_{k}^{\otimes (N-1)} \otimes X \otimes \id_{d}^{\otimes (N-1)}.
\label{eq:Extensions-X-KN}
\end{align}
where $(N+k-1)$-tuples $\vec{i} \equiv (a_{2} , \ldots , a_{k} , i_{1} , \ldots , i_{N})$, $\vec{j} \equiv (a'_{2} , \ldots , a'_{k} , j_{1} , \ldots , j_{N})$, and it is clear that the index-map between $\rho_{\vec{i} , \vec{j}} \in \mathbb{M}_{d^{N+1} \times d^{N+1}}(\mathbb{C})$ and $\rho_{i_{0} i_{1} \cdots i_{N} , j_{0} j_{1} \cdots j_{N}} \in \mathbb{M}_{d^{N+1} \times d^{N+1}}(\mathbb{C})$ is one-to-one.

This dualization method aims to apply Schur-Weyl duality and implement Schur transform in the following subsection.
Alternative approaches to addressing $\bar{U} \otimes U$ symmetry exist. For instance, implementing partial transpose on $\rho_{k}$'s Alice's system converts $\bar{U} \otimes U$ to $U \otimes U$ through $\Pi_{k} \otimes X \to \tau_{AB} \otimes X^{\tp}$. Alternatively, representation theory of Brauer algebra provides a framework to linear programming with $\bar{U}^{\otimes p} \otimes U^{\otimes q}$ symmetry \cite{Grinko:2022mpd}. In this paper, we adopt the dualization method as the primary economic strategy; analyses based on other methods are deferred to future studies.

%%%%%%%%

%------------------------------------------------

%%%%%%%%

\subsection{Twirling on feasible states}  \label{subsec:twirling-Uk}

Now we utilize the auxiliary unitary to implement the symmetry reduction. Having converted the symmetry from $\bar{U} \otimes U$ to $U^{\otimes k}$, we use Schur transform and twirling operation to block-diagonalize $\rho_{k,N}$ into $\lambda$-blocks as mentioned in Introduction. The auxiliary unitary on each $k$-extended $\mathcal{H}$ is $U \otimes \id_{d}$ with shorthand $U \equiv U \otimes \id_{d}$ if no confusion. The objective function is invariant under twirling due to the property $X=U^{\dagger} X U$,
\begin{align}
\tr ( X \rho )
=
\tr ( U^{\dagger} X U \rho )
=
\tr ( X U \rho U^{\dagger} )
=
\tr ( X \int U \rho U^{\dagger} dU ),
\end{align}
The goal of this subsection is to present the following theorem.
\begin{theorem}[Auxiliary unitary twirling] \label{thm:Twirling-PartialUnitary}
Unitary twirling on $\rho_{k,N}$'s auxiliary produces below twirled state (up to the isomorphism $\ket{ p_{\lambda} , q_{\lambda} }\cong \ket{q_{\lambda} , p_{\lambda}}$),
\begin{align}
\mathcal{T}_{U}[\rho_{k,N}]
&=
\int_{\mathrm{U}(k)}
U^{\otimes (N+k-1)}
\rho_{k,N}
{U^{\dagger}}^{\otimes (N+k-1)} dU
\cong
\bigoplus_{\lambda \vdash_{k} (N+k-1)} w_{\lambda} \frac{\id_{\mathbb{U}_{k, \lambda}}}{\dim \mathbb{U}_{k, \lambda}} \otimes \rho_{\lambda},
\label{eq:thm:Twirling-AuxiliaryUnitary}
\end{align}
where $\lambda \vdash_{k} (N+k-1)$ are Young diagrams having $N+k-1$ boxes and having at most $k$ rows. The nonnegative numbers $\{ w_{\lambda} \}_{\lambda \vdash_{k} (N+k-1)}$ satisfies $\sum_{\lambda \vdash_{k} (N+k-1)} w_{\lambda}=1$, and associates with a matrix $\rho_{\lambda} \in \mathbb{M}_{(\dim \mathbb{Y}_{\lambda} \times d^{N+1}) \times (\dim \mathbb{Y}_{\lambda} \times d^{N+1})}(\mathbb{C})_{+}$ with unit trace, with below form under Schur basis,
\begin{align}
\rho_{\lambda}=\sum_{p_{\lambda} , p'_{\lambda} } \ketbra{p_{\lambda}}{p'_{\lambda}} \otimes \rho_{p_{\lambda} , p'_{\lambda} },
\ \text{where} \
\rho_{p_{\lambda} , p'_{\lambda} }
\in \mathbb{M}_{d^{N+1} \times d^{N+1}}(\mathbb{C}).
\label{eq:YoungDiagramState}
\end{align}
On the other hand, the $X_{k,N}$ under the Schur basis can be correspondingly written into
\begin{align}
X_{k,N}
\cong
\bigoplus_{\lambda \vdash_{k} (N+k-1)}
\id_{\mathbb{U}_{\lambda}}
\otimes
\mathbb{P}_{\mathbb{Y}_{\lambda / (1^{k})}}
\otimes X_{(N)}
, \ \text{where} \
X_{(N)}=X \otimes \id_{d}^{\otimes (N-1)}.
\label{eq:YoungDiagramOperator}
\end{align}
This form is immediately obtained by the Littlewood-Richardson rule. The $\mathbb{P}_{\mathbb{Y}_{\lambda / (1^{k})}}$ is the projector of skew representation $\mathbb{Y}_{\lambda / (1^{k})}$ \cite{ceccherini2010RepsSym} embedding in $\mathbb{Y}_{\lambda}$ that amounts to selecting the standard Young tableaux whose $1$ to $k$ boxes are aligned in the first column.
\end{theorem}
\begin{proof}
Denote $(N+k-1)$-tuples by $\vec{i} \equiv (a_{2} , \cdots , a_{k} , i_{1} , \cdots , i_{N})$. By adding matrix indices, the Schur transform $T$ is expressed in terms of $T_{\lambda q_{\lambda} p_{\lambda} ,\vec{i}}$ as below,
\begin{align}
T
=
\sum_{\lambda \vdash_{k} (N+1)} \sum_{q_{\lambda}=1}^{\dim \mathbb{U}_{\lambda}} \sum_{p_{\mu}=1}^{\dim \mathbb{Y}_{\lambda}}
\sum_{\vec{i} \in [k]^{N+1} }
T_{\lambda p_{\lambda} q_{\lambda} , \vec{i} } \ketbra{ \lambda , p_{\lambda} , q_{\lambda} }{\vec{i}}.
\end{align}
Likewise, add matrix indices into $U^{\otimes (N+k-1)}$ and express it as $(U^{\otimes (N+k-1)})_{\vec{i} , \vec{j}}$. Then we consider the unitary conjugation on the extended state $\rho_{k,N}$,
\begin{align}
&U^{\otimes (N+k-1)} \rho_{k,N} {U^{\dagger}}^{\otimes (N+k-1)}
\nonumber \\
&=
\sum_{\lambda , \lambda' \vdash_{k} (N+k-1)} \sum_{\text{repeat $p,q$-indices}}
\ketbra{\lambda p_{\lambda} q_{\lambda}}{\lambda' p'_{\lambda'} q'_{\lambda'}}
\otimes
U^{\lambda}_{q_{\lambda} \tilde{q}_{\lambda}}
\rho_{\lambda p_{\lambda} \tilde{q}_{\lambda} , \lambda' p'_{\lambda'} \tilde{q}'_{\lambda'} }
\bar{U}^{\lambda'}_{q'_{\lambda'} \tilde{q}'_{\lambda'}},
\label{eq:SchurTwirling-kblock}
\end{align}
where the block matrix $\rho_{\lambda p_{\lambda} \tilde{q}_{\lambda}, \lambda' p'_{\lambda'} \tilde{q}'_{\lambda'} } \in \mathbb{M}_{d^{N+1} \times d^{N+1}}(\mathbb{C})$ is defined by
\begin{align}
\rho_{\lambda p_{\lambda} \tilde{q}_{\lambda}, \lambda' p'_{\lambda'} \tilde{q}'_{\lambda'} }
=
\sum_{\vec{i} , \vec{j}}
T_{\lambda p_{\lambda} \tilde{q}_{\lambda} , \vec{i} } \rho_{\vec{i}, \vec{j}}
T^{-1}_{\vec{j} , \lambda' p'_{\lambda'} \tilde{q}'_{\lambda'} }.
\end{align}
The unitary twirling equalises the pairs $(\lambda, \lambda')$ and $(\tilde{q}_{\lambda}, \tilde{q}'_{\lambda'})$ by Peter-Weyl theorem
\begin{align}
\int_{\mathrm{U}(k)} dU (U_{\lambda} )_{ab} ( \bar{U}_{\lambda'} )_{a'b'}=\frac{1}{\dim{U_{k, \lambda}}} \delta_{\lambda, \lambda'} \delta_{aa'} \delta_{bb'},
\end{align}
leading to the diagonal-block form
\begin{align}
\mathcal{T}_{U}[\rho_{k,N}]
&=
\int_{\mathrm{U}(k)}
U^{\otimes (N+k-1)}
\rho_{k,N}
{U^{\dagger}}^{\otimes (N+k-1)} dU
\\
&=
\sum_{\lambda \vdash_{k} (N+k-1)} \sum_{q_{\lambda} , q'_{\lambda} } \sum_{p_{\lambda} , p'_{\lambda} }
\frac{1}{\dim{U_{k, \lambda}}} \ketbra{\lambda p_{\lambda} q_{\lambda}}{\lambda p'_{\lambda} q_{\lambda}}
\otimes
\rho_{\lambda p_{\lambda} q'_{\lambda} , \lambda p'_{\lambda} q'_{\lambda} }.
\label{eq:Twirling-AuxiliaryUnitary-Blocks}
\end{align}
Note that $\id_{\mathbb{U}_{k, \lambda}}=\sum_{q_{\lambda}} \ketbra{q_{\lambda}}{q_{\lambda}}$ and $\rho_{p_{\lambda} , p'_{\lambda} }=\sum_{q'_{\lambda}} \rho_{\lambda p_{\lambda} q'_{\lambda} , \lambda p'_{\lambda} q'_{\lambda} }$, then by permuting the order of convention $\ket{\lambda p_{\lambda} q_{\lambda}} \cong \ket{\lambda q_{\lambda} p_{\lambda}}$, we could write the form of each $\lambda$-block,
\begin{align}
\frac{\id_{\mathbb{U}_{k, \lambda}}}{\dim{U_{k, \lambda}}} \otimes \rho_{\lambda}
\cong
\int_{\mathrm{U}(k)} {U}_{\lambda} (T \rho_{k,N} T^{-1})_{\lambda} U^{\dagger}_{\lambda} dU,
\ \text{where }
\rho_{\lambda}=\sum_{p_{\lambda} , p'_{\lambda} } \ketbra{p_{\lambda}}{p'_{\lambda}} \otimes \rho_{p_{\lambda} , p'_{\lambda} }.
\end{align}
The $w_{\lambda}$ are then defined as nonnegative numbers since every $\rho_{\lambda} \geq 0$.
\end{proof}

Recalling Eq.\eqref{eq:Blocks-Nextension-states-dual}, the presence of $\epsilon_{a_{2} \cdots a_{k} i_{0}}$ fixes the first $k-1$ data $a_{2} \cdots a_{k}$ (Alice's data) into Young diagram $(1^{k-1})$. Following Littlewood-Richardson rule, at level $N=1$ the decomposition to $\mathrm{Alt}^{k-1} \mathbb{C}^{k} \otimes \mathbb{C}^{k}$ produces standard Young tableaux $a_{(1^{k})}$ and $s_{(2,1^{k-2})}$ as below,
\begin{align}
\ytableausetup{mathmode, boxsize=normal}
\epsilon_{a_{2} \cdots a_{k} i_{0}} \to
  \begin{ytableau}
   {\scriptstyle 1} \\
   {\scriptstyle 2} \\
   {\scriptstyle \cdots} \\
   {\scriptstyle k-1}
  \end{ytableau}
  , \qquad
    \begin{ytableau}
   {\scriptstyle 1} \\
   {\scriptstyle 2} \\
   {\scriptstyle \cdots} \\
   {\scriptstyle k-1}
  \end{ytableau} \otimes
    \begin{ytableau}
   {\scriptstyle k}
  \end{ytableau}
  =
\underbrace{
  \begin{ytableau}
   {\scriptstyle 1} \\
   {\scriptstyle 2} \\
   {\scriptstyle \cdots} \\
   {\scriptstyle k-1} \\
   {\scriptstyle k}
  \end{ytableau}
  }_{\equiv a_{(1^{k})}}
  \oplus
  \underbrace{
  \begin{ytableau}
   {\scriptstyle 1} & {\scriptstyle k} \\
   {\scriptstyle 2} \\
   {\scriptstyle \cdots} \\
   {\scriptstyle k-1}
  \end{ytableau}
  }_{\equiv s_{(2,1^{k-2})}}.
\label{eq:Coupling-N=1}
\end{align}
\begin{corollary}
At level $N=1$, twirled state $\mathcal{T}_{U}[\rho_{k,1}]$ is decomposed into the blocks associated with Young diagrams $(1^{k})$ and $(2,1^{k-2})$,
\begin{align}
\mathcal{T}_{U}[\rho_{k,1}]
&=
w_{(1^{k})}
\id_{\mathbb{U}_{k, (1^{k})}} \otimes
\ketbra{a_{(1^{k})}}{a_{(1^{k})}} \otimes \rho_{a_{(1^{k})},a_{(1^{k})}}
\nonumber \\
&+
w_{(2,1^{k-2})}
\frac{\id_{\mathbb{U}_{k, (2,1^{k-2})}}}{k^2-1}
\ketbra{s_{(2,1^{k-2})}}{s_{(2,1^{k-2})}} \otimes \rho_{s_{(2,1^{k-2})},s_{(2,1^{k-2})}},
\label{eq:TwirledStates-N=1}
\end{align}
where $w_{(1^{k})} \geq 0$ and $w_{(2,1^{k-2})} \geq 0$ satisfy probability constraint $w_{(1^{k})}+w_{(2,1^{k-2})}=1$.
\end{corollary}

We could illustrate $\mathcal{T}_{U}[\rho_{k,N}]$'s and $X_{k,N}$'s decompositions at $N>1$ level under Schur basis as
\begin{eqnarray}
\mathcal{T}_{U}[\rho_{k,N}]=
    \begin{tikzpicture}[baseline={-0.5ex},mymatrixenv]
        \matrix [mymatrix,inner sep=4pt] (m)  
        {
    \tikzmarkin[kwad=style purple]{lambda} \rho_{\lambda} \tikzmarkend{lambda} & & & \\
     & \tikzmarkin[kwad=style purple]{lambdap} \rho_{\lambda'} \tikzmarkend{lambdap} &  \\
    & & \tikzmarkin[kwad=style purple]{lambdapp} \rho_{\lambda''} \tikzmarkend{lambdapp} & \\
    & & & \ddots
    \\
    };
    \end{tikzpicture}
    , \qquad
    X_{k,N}=
    \begin{tikzpicture}[baseline={-1ex},mymatrixenv]
        \matrix [mymatrix,inner sep=4pt] (p)  
        {
    \tikzmarkin[kwad=style purple]{lambdak} X_{\lambda} & & & \\
     & \ddots & & \\
    & & X_{\lambda''} \tikzmarkend{lambdak} & \\
    & & & \\ 
    };
        % Braces     
        \mymatrixbracetop{1}{3}{$\{ \lambda | \ell(\lambda)=k\}$}
    \end{tikzpicture}.
    \label{eq:DiagramBlocks}
\end{eqnarray}
Call each block in $\mathcal{T}_{U}[\rho_{k,N}]$ a diagram-block, or $\lambda$-block if Young diagram $\lambda$ is specified. Since $X_{k,N}$ is also block-diagonal with respect to diagrams, denote $X_{\lambda}$ the corresponding $\lambda$-block in $X_{k,N}$. The $\mathbb{P}_{\mathbb{Y}_{\lambda / (1^{k})}}$ in Eq.\eqref{eq:YoungDiagramOperator} implies that $X_{k,N}$ is only support in the diagrams that have rows equal to $k$, that is, $\ell(\lambda)=k$.

To sum up, we have shown that the twirled state $\mathcal{T}_{U}[\rho_{k,N}]$ on which SDP reduction is based, is block diagonal with diagram-blocks. A diagram block itself consists of block matrices labeled by joint indices $(p_{\lambda},p'_{\lambda})$, say, $\rho_{p_{\lambda} , p'_{\lambda} }$ and $X_{p_{\lambda} , p'_{\lambda} }$, which will be called tableau-labeled matrices.
The $\mathbb{P}_{\mathbb{Y}_{\lambda / (1^{k})}}$ in Eq.\eqref{eq:YoungDiagramOperator} also implies that $X_{p_{\lambda} , p'_{\lambda} }$ vanishes unless both $p_{\lambda}$ and $p'_{\lambda}$ correspond to standard Young tableaux whose first column is filled with $1$ to $k$. The next section will take a closer look at tableau-labeled matrices.

%%%%%%%%

%------------------------------------------------

%%%%%%%%

\section{Permutational Symmetry} \label{sec:PermutationalSym}

This section aims to have a closer look at the internal structure of diagram-blocks. Each diagram block is closed under permutation operation $S_{N+k-1}$, thus accessible to make $S_{N}$ Bose symmetry as SDP constraints. We begin with the following theorem and then explain it in the following subsections.
\begin{theorem}[SDP with BSE constraint] \label{thm:KBPSDP-SR-BSE}
The $\mathsf{SDP}_{k,N}(X)=\mathsf{SDP}_{k,N}^{\sym}(X)$ where $\mathsf{SDP}_{k,N}^{\sym}(X)$ is $\mathsf{SDP}_{k,N}(X)$'s reduction defined as follows,
\begin{align}
\mathsf{SDP}_{k,N}^{\sym}(X)
&:=
\min_{\{\rho_{\lambda} \in \mathrm{Pos}(\mathbb{C}^{d_{\lambda}} \otimes (\mathbb{C}^{d})^{\otimes (N+1)}), \lambda \vdash_{k} (N+k-1)\}} 
\tr [ ( \mathbb{P}_{\mathbb{Y}_{\lambda / (1^{k})}} \otimes X_{(N)} ) \rho_{\lambda} ], \\
\text{subject to } \:
&
\Delta_{\lambda}(\tau) \rho_{\lambda} = \rho_{\lambda}
, \
\forall \tau \in \mr{Cox}_{N}
, \ \text{and }
\tr\rho_{\lambda}=1.
\nonumber
\end{align}
The notations here are:
\begin{enumerate}
\item
For a given Young diagram $\lambda=(\lambda_{1} , \ldots , \lambda_{k})$, denote $\mathrm{SYT}_{\lambda / (1^{k})}$ and $\mathrm{SYT}_{\lambda / (2,1^{k-2})}$ the sets of standard Young tableaux based on skew shape $\lambda / (1^{k})$ and $\lambda / (2,1^{k-2})$, respectively. These standard Young tableaux of $\mathrm{SYT}_{\lambda / (1^{k})}$ (respect $\mathrm{SYT}_{\lambda / (2,1^{k-2})}$) span the subspace $\mathbb{Y}_{\lambda / (1^{k})}$ (respect $\mathbb{Y}_{\lambda / (2,1^{k-2})}$) of $\mathbb{Y}_{\lambda}$. Denote $\mathbb{P}_{\mathbb{Y}_{\lambda / (1^{k})}}$ and $\id_{\mathbb{Y}_{\lambda / (2,1^{k-2})}}$ the respect projectors of the subspaces.

\item
Denote $\mathrm{Pos}(V)$ the set of positive definite matrices with respect to vector space $V$.
The size of $\rho_{\lambda}$ is $d^{N+1} \times O(k^{N+1} (N-1)^{-\frac{k^2+k-2}{4}})$ in Big O notation, where the block size, defined as the ratio $\mathrm{size}(rho_{\lambda})/d^{N+1}$ with $d^{N+1}$ the size of tableau-labeled matrices, is $d_{\lambda}=\dim \mathbb{Y}_{\lambda / (1^{k})}+f^{\lambda / (2,1^{k-1})}$ Eq.\eqref{eq:formulas-DiagramBlock}. We have $d_{\lambda} \sim O(k^{N} (N-1)^{-\frac{k^2+k-2}{4}})$.

\item
$\mr{Cox}_{N}=\{ (j,j+1) \in S_{N} : k \leq j \leq N+k-2 \}$ is the set of Coxeter generators of $S_{N}$;

\item
$\Delta_{\lambda} : S_N \to \mathrm{End}(\mathbb{C}^{d_{\lambda}}
 \otimes (\mathbb{C}^{d})^{\otimes (N+1)})$ is induced from
$\Delta_{B} : S_N \to \mathrm{U}((\mathbb{C}^{k}
 \otimes \mathbb{C}^{d})^{\otimes N})$ with $\id_{A} \otimes \Delta_{B}(\pi) \mapsto U_{\pi}^{\lambda} \otimes U_{\pi}$;
 
\item There are at most $(N-1) d_{\lambda}^2 \times d^{N+1}$ many of constraints.
\end{enumerate}
\end{theorem}
Points 1 and 2 will be explained in Subsection~\ref{subsec:Asymptotic Diagram-blockSize}, and Points 3-5 in Subsection~\ref{subsec:symmetric-reduction-permutation}.

%%%%%%%%

%------------------------------------------------

%%%%%%%%

\subsection{Permutation constraints arisen from $N$-BSE symmetry} \label{subsec:symmetric-reduction-permutation}

A permutation $\pi \in S_{N}$ only acts on $\rho_{\lambda}$ and the left-action is given by $\Delta_{\lambda}: S_{N} \to \mathbb{Y}_{\lambda} \otimes \mathrm{U}( (\mathbb{C}^{d})^{\otimes N})$,
\begin{align}
\Delta_{\lambda}(\pi) \rho_{\lambda}
=
\sum_{p_{\lambda} , p'_{\lambda} }
\pi \ketbra{p_{\lambda}}{p'_{\lambda}} \otimes \pi \rho_{p_{\lambda} , p'_{\lambda}}
=
\sum_{p_{\lambda} , p'_{\lambda} , p''_{\lambda}}
\ketbra{p''_{\lambda}}{p'_{\lambda}} \otimes \pi_{p''_{\lambda} p_{\lambda}} (\pi \rho_{p_{\lambda} , p'_{\lambda}} ),
\label{eq:PermutatingYoungDiagramState}
\end{align}
where $\rho_{p_{\lambda} p'_{\lambda}} \in \mathbb{M}_{d^{N+1} \times d^{N+1}}(\mathbb{C})$ denotes a tableau-labeled matrix labeled by pair $(p_{\lambda},p'_{\lambda})$, and $\pi_{p''_{\lambda} p_{\lambda}}$ is the irreducible representation matrix of $\pi$ associative with $\mathbb{Y}_{\lambda}$ decomposing auxiliary space, meanwhile $\pi$'s action on $(\mathbb{C}^{d})^{\otimes N}$ is defined in the natural way as $\pi \ket{ e_{1} \ldots e_{N} }=\ket{e_{\pi(1)} \ldots e_{\pi(N)}}$.

Above equation implies below constraints with respect to permutation invariance,
\begin{align}
\pi^{-1} \rho_{p_{\lambda} , p'_{\lambda}}
=
\sum_{p''_{\lambda}}
\pi_{p_{\lambda} p''_{\lambda}} \rho_{p''_{\lambda} , p'_{\lambda}}.
\label{eq:PermutationInvariance-Constraints}
\end{align}
This equation includes the situation that $\Delta_{B}(\pi)$ acts from $\rho_{\lambda}$'s right side due to $\rho_{p_{\lambda},p'_{\lambda}}^{\dagger}=\rho_{p'_{\lambda},p_{\lambda}}$.

Let us introduce some notations: for Young diagram $\lambda \vdash_{k} (N+k-1)$ having $k$ rows, we can classify its standard Young tableaux into
\begin{align}
&
\mr{SYT}_{\lambda}^{a}
:=
\mr{SYT}_{\lambda / (1^{k}) }
=
\{ a_{\lambda} \in \mr{SYT}_{\lambda} : a_{\lambda}(i,1)=i , 1 \leq i \leq k \},
\\
&
\mr{SYT}_{\lambda}^{s}
:=
\mr{SYT}_{\lambda / (2,1^{k-2})}
=
\{ s_{\lambda} \in \mr{SYT}_{\lambda} : s_{\lambda}(i,1)=i , 1 \leq i \leq k-1 , \text{ and } w(1,2)=k \},
\\
&
\mr{SYT}_{\lambda}^{m} = \mr{SYT}_{\lambda} \setminus (\mr{SYT}_{\lambda}^{a} \sqcup \mr{SYT}_{\lambda}^{s}),
\end{align}
where $(i,j)$ coordinates the box in $\lambda$ at the $i$th row and the $j$th column. The illustration below is for the $\mr{SYT}_{(6,4,3,2,1)}^{a}$, $\mr{SYT}_{(6,4,3,2,1)}^{s}$, and $\mr{SYT}_{(6,4,3,2,1)}^{m}$ (from left to right),
\begin{align}
\ytableausetup{mathmode, boxsize=normal}
  \begin{ytableau}
   {\scriptstyle 1} & \quad & \quad & \quad & \quad & \quad \\
   {\scriptstyle 2} & \quad & \quad & \quad \\
   {\scriptstyle \vdots} & \quad & \quad \\
   {\scriptstyle k-1} & \quad \\
   {\scriptstyle k}
  \end{ytableau}
  \ \text{or} \
  \begin{ytableau}
   {\scriptstyle 1} & {\scriptstyle k} & \quad & \quad & \quad & \quad \\
   {\scriptstyle 2} & \quad & \quad & \quad \\
   {\scriptstyle \vdots} & \quad & \quad \\
   {\scriptstyle k-1} & \quad \\
   {\scriptstyle *}
  \end{ytableau},
    \ \text{or} \
  \begin{ytableau}
   {\scriptstyle 1} & \quad & \quad & \quad & \quad & \quad \\
   {\scriptstyle *} & \quad & \quad & \quad \\
   {\scriptstyle \vdots} & \quad & \quad \\
   {\scriptstyle *} & \quad \\
   {\scriptstyle *}
  \end{ytableau}.
\end{align}
A more explicit example for $k=3$ and $N=4$ under $\lambda=(3,2,1)$ is displayed below,
\begin{align*}
&
\{ a_{\lambda} \}
=
\bigg\{
  \begin{ytableau}
   {\scriptstyle 1} & {\scriptstyle 4} & {\scriptstyle 6} \\
   {\scriptstyle 2} & {\scriptstyle 5} \\
   {\scriptstyle 3}
  \end{ytableau}, \quad
    \begin{ytableau}
   {\scriptstyle 1} & {\scriptstyle 4} & {\scriptstyle 5} \\
   {\scriptstyle 2} & {\scriptstyle 6} \\
   {\scriptstyle 3}
  \end{ytableau}
  \bigg\},
\\
&
  \{ s_{\lambda} \}
=
\bigg\{
  \begin{ytableau}
   {\scriptstyle 1} & {\scriptstyle 2} & {\scriptstyle 6} \\
   {\scriptstyle 3} & {\scriptstyle 5} \\
   {\scriptstyle 4}
  \end{ytableau}, \quad
    \begin{ytableau}
   {\scriptstyle 1} & {\scriptstyle 2} & {\scriptstyle 5} \\
   {\scriptstyle 3} & {\scriptstyle 6} \\
   {\scriptstyle 4}
  \end{ytableau},
\quad
 \begin{ytableau}
   {\scriptstyle 1} & {\scriptstyle 2} & {\scriptstyle 6} \\
   {\scriptstyle 3} & {\scriptstyle 4} \\
   {\scriptstyle 5}
\end{ytableau}
\quad
\begin{ytableau}
   {\scriptstyle 1} & {\scriptstyle 2} & {\scriptstyle 4} \\
   {\scriptstyle 3} & {\scriptstyle 6} \\
   {\scriptstyle 5}
  \end{ytableau},
  \quad
  \begin{ytableau}
   {\scriptstyle 1} & {\scriptstyle 2} & {\scriptstyle 5} \\
   {\scriptstyle 3} & {\scriptstyle 4} \\
   {\scriptstyle 6}
  \end{ytableau},
  \quad
  \begin{ytableau}
   {\scriptstyle 1} & {\scriptstyle 2} & {\scriptstyle 4} \\
   {\scriptstyle 3} & {\scriptstyle 5} \\
   {\scriptstyle 6}
  \end{ytableau}
 \bigg\}, 
\\
&
\{ m_{\lambda} \}
=
\bigg\{
  \begin{ytableau}
   {\scriptstyle 1} & {\scriptstyle 2} & {\scriptstyle 3} \\
   {\scriptstyle 4} & {\scriptstyle 6} \\
   {\scriptstyle 5}
  \end{ytableau}, \quad
    \begin{ytableau}
   {\scriptstyle 1} & {\scriptstyle 2} & {\scriptstyle 3} \\
   {\scriptstyle 4} & {\scriptstyle 5} \\
   {\scriptstyle 6}
  \end{ytableau}, \quad
  \begin{ytableau}
   {\scriptstyle 1} & {\scriptstyle 3} & {\scriptstyle 6} \\
   {\scriptstyle 2} & {\scriptstyle 5} \\
   {\scriptstyle 4}
  \end{ytableau}, \quad \cdots
  \bigg\}.
\end{align*}

We can adopt column-lexicographic order $\preceq$ to label tableaux, so $\mr{SYT}_{\lambda}^{a} \preceq \mr{SYT}_{\lambda}^{s} \preceq \mr{SYT}_{\lambda}^{m}$ such that $p_{\lambda}=1,\ldots,\dim \mathbb{Y}_{\lambda}$. The $\rho_{\lambda}$ and $X_{\lambda}$ can be expressed as the following the block structures,
\begin{eqnarray}
&
\rho_{\lambda}=
    \begin{tikzpicture}[baseline={-0.5ex},mymatrixenv]
        \matrix [mymatrix,inner sep=4pt] (m)  
        {
    \tikzmarkin[kwad=style green]{aa} \rho_{a_{0},a_{0}}  & \cdots  & \rho_{a_{0},a_{1}}  & \tikzmarkin[kwad=style orange]{ncn} \rho_{a_{0},s_{0}} & \cdots & \rho_{a_{0},s_{1}} & \tikzmarkin[kwad=style cyan]{rmne} \rho_{a_{0},m_{0}} & \cdots & \rho_{a_{0},m_{1}} \\
    \vdots & \ddots & \vdots & \vdots & \ddots & \vdots & \vdots & \ddots & \vdots \\
    \rho_{a_{1},a_{0}}  & \cdots  & \rho_{a_{1},a_{1}} \tikzmarkend{aa} & \rho_{a_{1},s_{0}} & \cdots & \rho_{a_{1},s_{1}} \tikzmarkend{ncn} & \rho_{a_{1},m_{0}} & \cdots & \rho_{a_{1},m_{1}} \tikzmarkend{rmne} \\
\tikzmarkin[kwad=style orange]{ncw} \rho_{s_{0},a_{0}}  & \cdots  & \rho_{s_{0},a_{1}}  & \tikzmarkin[kwad=style orange]{ncc} \rho_{s_{0},s_{0}} & \cdots & \rho_{s_{0},s_{1}} & \tikzmarkin[kwad=style cyan]{rmec} \rho_{s_{0},m_{0}} & \cdots & \rho_{s_{0},m_{1}} \\
    \vdots & \ddots & \vdots & \vdots & \ddots & \vdots & \vdots & \ddots & \vdots \\
    \rho_{s_{1},a_{0}}  & \cdots  & \rho_{s_{1},a_{1}} \tikzmarkend{ncw} & \rho_{s_{1},s_{0}} & \cdots & \rho_{s_{1},s_{1}} \tikzmarkend{ncc} & \rho_{s_{1},m_{0}} & \cdots & \rho_{s_{1},m_{1}} \tikzmarkend{rmec} \\
\tikzmarkin[kwad=style cyan]{rmsw} \rho_{m_{0},a_{0}}  & \cdots  & \rho_{m_{0},a_{1}} & \tikzmarkin[kwad=style cyan]{rmsc} \rho_{m_{0},s_{0}} & \cdots & \rho_{m_{0},s_{1}} & \tikzmarkin[kwad=style cyan]{rmc} \rho_{m_{0},m_{0}} & \cdots & \rho_{m_{0},m_{1}} \\
    \vdots & \ddots & \vdots & \vdots & \ddots & \vdots & \vdots & \ddots & \vdots \\
    \rho_{m_{1},a_{0}}  & \cdots  & \rho_{m_{1},a_{1}} \tikzmarkend{rmsw} & \rho_{m_{1},s_{0}} & \cdots & \rho_{m_{1},s_{1}} \tikzmarkend{rmsc} & \rho_{m_{1},m_{0}} & \cdots & \rho_{m_{1},m_{1}} \tikzmarkend{rmc} \\
    };  
        \mymatrixbraceright{1}{3}{$\{ a_{\lambda} \}$}
        \mymatrixbraceright{4}{6}{$\{ s_{\lambda} \}$}
        \mymatrixbraceright{7}{9}{$\{ m_{\lambda} \}$}
        \mymatrixbracetop{1}{3}{$\{ a'_{\lambda} \} $}
        \mymatrixbracetop{4}{6}{$\{ s'_{\lambda} \}$}
        \mymatrixbracetop{7}{9}{$\{ m'_{\lambda} \}$}
    \end{tikzpicture},
    \label{eq:Tableau-labeledMatrices-rho}
    \\
&
X_{\lambda}=
    \begin{tikzpicture}[baseline={-0.5ex},mymatrixenv]
        \matrix [mymatrix,inner sep=4pt] (m)  
        {
    \tikzmarkin[kwad=style green]{aax} X  &  &  &  &  & & & & \\
    & \ddots & & & & & & & \\
    & & X \tikzmarkend{aax} &  & & &  &  &  \\
    & & &  & & &  &  &  \\
    & & &  & & &  &  &  \\
    & & &  & & &  &  &  \\ 
    & & &  & & &  &  &  \\ 
    & & &  & & &  &  &  \\ 
    & & &  & & &  &  & \\
    };  
        \mymatrixbraceright{1}{3}{$\{ a_{\lambda} \}$}
        \mymatrixbraceright{4}{6}{$\{ s_{\lambda} \}$}
        \mymatrixbraceright{7}{9}{$\{ m_{\lambda} \}$}
        \mymatrixbracetop{1}{3}{$\{ a'_{\lambda} \} $}
        \mymatrixbracetop{4}{6}{$\{ s'_{\lambda} \}$}
        \mymatrixbracetop{7}{9}{$\{ m'_{\lambda} \}$}
    \end{tikzpicture},
    \label{eq:Tableau-labeledMatrices-X}
\end{eqnarray}
where $\rho_{a_{\lambda} a'_{\lambda}} , \ldots , \rho_{m_{\lambda} m'_{\lambda}}$ are tableau-labeled matrices with size $\mathbb{M}_{d^{N+1} \times d^{N+1}}(\mathbb{C})$ which are contained in blocks $\{ \rho_{a_{\lambda} a'_{\lambda}} \} , \ldots , \{ \rho_{m_{\lambda} m'_{\lambda}} \}$ respectively.
The $X_{p_{\lambda},p'_{\lambda}}$ is block-diagonal in the sense that $X_{p_{\lambda},p'_{\lambda}}=X_{p_{\lambda},p_{\lambda}}\delta_{p_{\lambda},p'_{\lambda}}$, and $X_{p_{\lambda} p'_{\lambda}} \neq 0$ only when $\ket{p_{\lambda}}=\Pi_{k} \ket{ p_{\lambda} } \neq 0$, or say, standard Young tableau whose first column is filled by $1$ to $k$.
And $X_{p_{\lambda},p_{\lambda}} \neq 0$ only when $p_{\lambda}$ belongs to the left-hand-side class.

The indices labeling tableau-labeled matrices are $a_{\lambda}=a_{0} , \ldots , a_{1}$ of $| \mr{SYT}_{\lambda}^{a} |$, and $s_{\lambda}=s_{0} , \ldots , s_{1}$ of $| \mr{SYT}_{\lambda}^{s} |$, and $m_{\lambda}=m_{0} , \ldots , m_{1}$ of $| \mr{SYT}_{\lambda}^{m} |$, with respect to cardinalities
\begin{align}
| \mr{SYT}_{\lambda}^{a} |=\dim \mathbb{Y}_{\lambda / (1^{k})},
\quad
| \mr{SYT}_{\lambda}^{s} |=\dim \mathbb{Y}_{\lambda / (2,1^{k-2})},
\quad
| \mr{SYT}_{\lambda}^{m} |=\dim \mathbb{Y}_{\lambda}-| \mr{SYT}_{\lambda}^{a} |-| \mr{SYT}_{\lambda}^{s} |.
\label{eq:formulas-DiagramBlock}
\end{align}

The size of $\lambda$-block, denoted by $d_{\lambda}$, could be chosen smaller than
$\dim \mathbb{Y}_{\lambda}$,
\begin{align}
d_{\lambda}
=
| \mr{SYT}_{\lambda}^{a} |
+
| \mr{SYT}_{\lambda}^{s} |
=
\dim \mathbb{Y}_{\lambda / (1^{k})}+\dim \mathbb{Y}_{\lambda / (2,1^{k-1})}
\leq \dim \mathbb{Y}_{\lambda}
\label{eq:Size-diagram-blocks}
\end{align}
by setting $\{ \rho_{m_{\lambda}, p_{\lambda}} \}$ and $\{ \rho_{p_{\lambda}, m_{\lambda}} \}$ ({\color{cyan}{cyan}} tableau-labeled matrices) to zero.
Coupling $\mathrm{Alt}^{k-1} \mathbb{C}^{k}$ (Alice's auxiliary) with $\mathbb{C}^{k}$ ($B_{1}$'s auxiliary) is fixed as Eq.\eqref{eq:Coupling-N=1} which should be either of $\{ a_{\lambda} \}$ or of $\{ s_{\lambda} \}$.
Note that the irreducible representation matrix $\pi_{ \lambda } \in S_{N}$ has the form of
\begin{eqnarray}
\pi_{ \lambda }
=
    \begin{tikzpicture}[baseline={-0.5ex},mymatrixenv]
        \matrix [mymatrix,inner sep=4pt] (m)  
        {
    \tikzmarkin[kwad=style green]{aapi} \pi_{a_{0},a_{0}}  & \cdots  & \pi_{a_{0},a_{1}}  & \tikzmarkin[kwad=style orange]{ncnpi} \pi_{a_{0},s_{0}} & \cdots & \pi_{a_{0},s_{1}} & & & \\
    \vdots & \ddots & \vdots & \vdots & \ddots & \vdots & & & \\
    \pi_{a_{1},a_{0}}  & \cdots  & \pi_{a_{1},a_{1}} \tikzmarkend{aapi} & \pi_{a_{1},s_{0}} & \cdots & \pi_{a_{1},s_{1}} \tikzmarkend{ncnpi} & & & \\
\tikzmarkin[kwad=style orange]{ncwpi} \pi_{s_{0},a_{0}}  & \cdots  & \pi_{s_{0},a_{1}}  & \tikzmarkin[kwad=style orange]{nccpi} \pi_{s_{0},s_{0}} & \cdots & \pi_{s_{0},s_{1}} & & & \\
    \vdots & \ddots & \vdots & \vdots & \ddots & \vdots & & & \\
    \pi_{s_{1},a_{0}}  & \cdots  & \pi_{s_{1},a_{1}} \tikzmarkend{ncwpi} & \pi_{s_{1},s_{0}} & \cdots & \pi_{s_{1},s_{1}} \tikzmarkend{nccpi} & && \\
 & & & & &  & \tikzmarkin[kwad=style cyan]{rmcpi} \pi_{m_{0},m_{0}} & \cdots & \pi_{m_{0},m_{1}} \\
    & & & & & & \vdots & \ddots & \vdots \\
     & & & & & & \pi_{m_{1},m_{0}} & \cdots & \pi_{m_{1},m_{1}} \tikzmarkend{rmcpi} \\
    };
        \mymatrixbraceright{1}{3}{$\{ a_{\lambda} \}$}
        \mymatrixbraceright{4}{6}{$\{ s_{\lambda} \}$}
        \mymatrixbraceright{7}{9}{$\{ m_{\lambda} \}$}
        \mymatrixbracetop{1}{3}{$\{ a'_{\lambda} \} $}
        \mymatrixbracetop{4}{6}{$\{ s'_{\lambda} \}$}
        \mymatrixbracetop{7}{9}{$\{ m'_{\lambda} \}$}
    \end{tikzpicture},
\end{eqnarray}
because basis $\ket{m_{\lambda}}$ is by no mean being transformed to neither $\mr{SYT}_{\lambda}^{a}$ nor $\mr{SYT}_{\lambda}^{s}$ by $S_{N}$.
The $\pi$'s left-action on $\rho_{\lambda}$ is given by
\begin{align}
\Delta{\pi} \cdot \rho_{\lambda}
=
\left( \pi_{\lambda} \otimes \id_{d}^{\otimes (N+1)}
\right)
\cdot
\left(\id_{\mathbb{Y}_{\lambda}} \otimes \pi \right)
\rho_{\lambda},
  \label{eq:Permutation-YoungDiagramStates}
\end{align}
and could be illustrated as
\begin{eqnarray}
\Delta_{\lambda}(\pi) \cdot \rho_{\lambda}
&=\begin{tikzpicture}[baseline={-0.5ex},mymatrixenv]
        \matrix [mymatrix,inner sep=4pt] (m)  
        {
    \tikzmarkin[kwad=style green]{aapirhopi} \pi_{a_{0},a_{0}} \id & \cdots  & \pi_{a_{0},a_{1}} \id & \tikzmarkin[kwad=style orange]{ncnpirhopi} \pi_{a_{0},s_{0}} \id & \cdots & \pi_{a_{0},s_{1}} \id & & & \\
    \vdots & \ddots & \vdots & \vdots & \ddots & \vdots & & & \\
    \pi_{a_{1},a_{0}} \id & \cdots  & \pi_{a_{1},a_{1}} \id \tikzmarkend{aapirhopi} & \pi_{a_{1},s_{0}} \id & \cdots & \pi_{a_{1},s_{1}} \id \tikzmarkend{ncnpirhopi} & & & \\
\tikzmarkin[kwad=style orange]{ncwpirhopi} \pi_{s_{0},a_{0}} \id & \cdots  & \pi_{s_{0},a_{1}} \id & \tikzmarkin[kwad=style orange]{nccpirhopi} \pi_{s_{0},s_{0}} \id & \cdots & \pi_{s_{0},s_{1}} \id & & & \\
    \vdots & \ddots & \vdots & \vdots & \ddots & \vdots & & & \\
    \pi_{s_{1},a_{0}} \id & \cdots  & \pi_{s_{1},a_{1}} \id \tikzmarkend{ncwpirhopi} & \pi_{s_{1},s_{0}} \id & \cdots & \pi_{s_{1},s_{1}} \id \tikzmarkend{nccpirhopi} & && \\
 & & & & &  & \tikzmarkin[kwad=style cyan]{rmcpirhopi} \pi_{m_{0},m_{0}} \id & \cdots & \pi_{m_{0},m_{1}} \id \\
    & & & & & & \vdots & \ddots & \vdots \\
     & & & & & & \pi_{m_{0},m_{1}} \id & \cdots & \pi_{m_{1},m_{1}} \id \tikzmarkend{rmcpirhopi} \\
    };
    \end{tikzpicture}
\nonumber \\
&
\cdot
    \begin{tikzpicture}[baseline={-0.5ex},mymatrixenv]
        \matrix [mymatrix,inner sep=4pt] (m)  
        {
    \tikzmarkin[kwad=style green]{aapirhorho} \pi \rho_{a_{0},a_{0}}  & \cdots  & \pi \rho_{a_{0},a_{1}}  & \tikzmarkin[kwad=style orange]{ncnpirhorho} \pi \rho_{a_{0},s_{0}} & \cdots & \pi \rho_{a_{0},s_{1}} & \tikzmarkin[kwad=style cyan]{rmnepirhorho} \pi \rho_{a_{0},m_{0}} & \cdots & \pi \rho_{a_{0},m_{1}} \\
    \vdots & \ddots & \vdots & \vdots & \ddots & \vdots & \vdots & \ddots & \vdots \\
    \pi \rho_{a_{1},a_{0}}  & \cdots  & \pi \rho_{a_{1},a_{1}} \tikzmarkend{aapirhorho} & \pi \rho_{a_{1},s_{0}} & \cdots & \pi \rho_{a_{1},s_{1}} \tikzmarkend{ncnpirhorho} & \pi \rho_{a_{1},m_{0}} & \cdots & \pi \rho_{a_{1},m_{1}} \tikzmarkend{rmnepirhorho} \\
\tikzmarkin[kwad=style orange]{ncwpirhorho} \pi \rho_{s_{0},a_{0}}  & \cdots  & \pi \rho_{s_{0},a_{1}}  & \tikzmarkin[kwad=style orange]{nccpirhorho} \pi \rho_{s_{0},s_{0}} & \cdots & \pi \rho_{s_{0},s_{1}} & \tikzmarkin[kwad=style cyan]{rmecpirhorho} \pi \rho_{s_{0},m_{0}} & \cdots & \pi \rho_{s_{0},m_{1}} \\
    \vdots & \ddots & \vdots & \vdots & \ddots & \vdots & \vdots & \ddots & \vdots \\
   \pi  \rho_{s_{1},a_{0}}  & \cdots  & \pi \rho_{s_{1},a_{1}} \tikzmarkend{ncwpirhorho} & \pi \rho_{s_{1},s_{0}} & \cdots & \pi \rho_{s_{1},s_{1}} \tikzmarkend{nccpirhorho} & \pi \rho_{s_{1},m_{0}} & \cdots & \pi \rho_{s_{1},m_{1}} \tikzmarkend{rmecpirhorho} \\
\tikzmarkin[kwad=style cyan]{rmswpirhorho} \pi \rho_{m_{0},a_{0}}  & \cdots  & \pi \rho_{m_{0},a_{1}} & \tikzmarkin[kwad=style cyan]{rmscpirhorho} \pi \rho_{m_{0},s_{0}} & \cdots & \pi \rho_{m_{0},s_{1}} & \tikzmarkin[kwad=style cyan]{rmcpirhorho} \pi \rho_{m_{0},m_{0}} & \cdots & \pi \rho_{m_{0},m_{1}} \\
    \vdots & \ddots & \vdots & \vdots & \ddots & \vdots & \vdots & \ddots & \vdots \\
    \pi \rho_{m_{1},a_{0}}  & \cdots  & \pi \rho_{m_{1},a_{1}} \tikzmarkend{rmswpirhorho} & \pi \rho_{m_{1},s_{0}} & \cdots & \pi \rho_{m_{1},s_{1}} \tikzmarkend{rmscpirhorho} & \pi \rho_{m_{1},m_{0}} & \cdots & \pi \rho_{m_{1},m_{1}} \tikzmarkend{rmcpirhorho} \\
    };
    \end{tikzpicture}.
  \label{eq:Permutation-YoungDiagramStates-Illustration}
\end{eqnarray}
The {\color{cyan}{cyan}} tableau-labeled matrices $\{ \rho_{m_{\lambda}, p_{\lambda}} \}$ and $\{ \rho_{p_{\lambda}, m_{\lambda}} \}$ are removable since they are invariant subspace under the permutation thus setting them into zero matrices is still consistent with permutation constraints.
Hence, below simplification is admitted,
\begin{eqnarray}
&
\rho_{\lambda} \mapsto
\rho_{\lambda}=
    \begin{tikzpicture}[baseline={-0.5ex},mymatrixenv]
        \matrix [mymatrix,inner sep=4pt] (m)  
        {
    \tikzmarkin[kwad=style green]{aar} \rho_{a_{0},a_{0}}  & \cdots  & \rho_{a_{0},a_{1}}  & \tikzmarkin[kwad=style orange]{ncnr} \rho_{a_{0},s_{0}} & \cdots & \rho_{a_{0},s_{1}} \\
    \vdots & \ddots & \vdots & \vdots & \ddots & \vdots \\
    \rho_{a_{1},a_{0}}  & \cdots  & \rho_{a_{1},a_{1}} \tikzmarkend{aar} & \rho_{a_{1},s_{0}} & \cdots & \rho_{a_{1},s_{1}} \tikzmarkend{ncnr} \\
\tikzmarkin[kwad=style orange]{ncwr} \rho_{s_{0},a_{0}}  & \cdots  & \rho_{s_{0},a_{1}}  & \tikzmarkin[kwad=style orange]{nccr} \rho_{s_{0},s_{0}} & \cdots & \rho_{s_{0},s_{1}} \\
    \vdots & \ddots & \vdots & \vdots & \ddots & \vdots \\
    \rho_{s_{1},a_{0}}  & \cdots  & \rho_{s_{1},a_{1}} \tikzmarkend{ncwr} & \rho_{s_{1},s_{0}} & \cdots & \rho_{s_{1},s_{1}} \tikzmarkend{nccr} \\
    };  
        \mymatrixbraceright{1}{3}{$\{ a_{\lambda} \}$}
        \mymatrixbraceright{4}{6}{$\{ s_{\lambda} \}$}
        \mymatrixbracetop{1}{3}{$\{ a'_{\lambda} \} $}
        \mymatrixbracetop{4}{6}{$\{ s'_{\lambda} \}$}
    \end{tikzpicture},
    \label{eq:Tableau-labeledMatrices-rho-reduced}
\\
&
\pi_{ \lambda }
\mapsto
\pi_{ \lambda }
=
    \begin{tikzpicture}[baseline={-0.5ex},mymatrixenv]
        \matrix [mymatrix,inner sep=4pt] (m)  
        {
    \tikzmarkin[kwad=style green]{aapir} \pi_{a_{0},a_{0}}  & \cdots  & \pi_{a_{0},a_{1}}  & \tikzmarkin[kwad=style orange]{ncnpir} \pi_{a_{0},s_{0}} & \cdots & \pi_{a_{0},s_{1}} \\
    \vdots & \ddots & \vdots & \vdots & \ddots & \vdots \\
    \pi_{a_{1},a_{0}}  & \cdots  & \pi_{a_{1},a_{1}} \tikzmarkend{aapir} & \pi_{a_{1},s_{0}} & \cdots & \pi_{a_{1},s_{1}} \tikzmarkend{ncnpir} \\
\tikzmarkin[kwad=style orange]{ncwpir} \pi_{s_{0},a_{0}}  & \cdots  & \pi_{s_{0},a_{1}}  & \tikzmarkin[kwad=style orange]{nccpir} \pi_{s_{0},s_{0}} & \cdots & \pi_{s_{0},s_{1}} \\
    \vdots & \ddots & \vdots & \vdots & \ddots & \vdots \\
    \pi_{s_{1},a_{0}}  & \cdots  & \pi_{s_{1},a_{1}} \tikzmarkend{ncwpir} & \pi_{s_{1},s_{0}} & \cdots & \pi_{s_{1},s_{1}} \tikzmarkend{nccpir} \\
    };
        \mymatrixbraceright{1}{3}{$\{ a_{\lambda} \}$}
        \mymatrixbraceright{4}{6}{$\{ s_{\lambda} \}$}
        \mymatrixbracetop{1}{3}{$\{ a'_{\lambda} \} $}
        \mymatrixbracetop{4}{6}{$\{ s'_{\lambda} \}$}
    \end{tikzpicture},
\end{eqnarray}
This simplification also alleviates the number of scalar constraints:
The matrix equation Eq.\eqref{eq:Permutation-YoungDiagramStates} produces at most $(N-1) d_{\lambda}^2 d^{N+1}$ scalar constraints where $N-1$ is the number of Coxeter generators apart from $\id$, and $d_{\lambda}^2$ the square of $\lambda$-block's size, $d^{N+1}$ the size of tableau-labeled matrices.

%%%%%%%%

\subsection{Asymptotic diagram-block size} \label{subsec:Asymptotic Diagram-blockSize}

Consider Young diagram $\lambda=(\lambda_{1}, \ldots , \lambda_{k}) \vdash_{k} N+k-1$. Each $\lambda$-block contains tableau-labeled matrices of size $d^{N+1}$. So we should estimate $d_{\lambda}$ in Eq.\eqref{eq:Size-diagram-blocks}.
We compute the ratio $\dim \mathbb{Y}_{\lambda / (1^{k})}$ and $\dim \mathbb{Y}_{\lambda / (2,1^{k-2})}$, for the reason that the ratio should be expected to relate to the optimal value in the following manner:
\begin{align*}
\mathsf{SDP}_{k,N}^{\lambda}
\propto
\frac{
\dim \mathbb{Y}_{\lambda / (1^{k})}
}{
\dim \mathbb{Y}_{\lambda / (1^{k})} + \dim \mathbb{Y}_{\lambda / (2,1^{k-2})}
}, \ \text{as } N \to \infty,
\end{align*}
because of $\tr \rho_{\lambda}=1$ and the fact that only blocks in $\mathbb{Y}_{\lambda / (1^{k})}$ contribute to the objective function.
Interestingly, this ratio relates to the shifted Schur function \cite{Goulden1994superSchur,okounkov1996shifted,Borodin_Olshanski_2016},
\begin{align}
\frac{\dim \mathbb{Y}_{\lambda / \mu}}{\dim \mathbb{Y}_{\lambda}}
=
\frac{s_{\mu}^{*}(\lambda)}
{ | \lambda |^{\downarrow |\mu| }},
\label{eq:Ratio-skewSYT-shiftedSchur}
\end{align}
where $s_{\mu}^{*}$ denotes shifted Schur function with arguments $\lambda=(\lambda_{1}, \lambda_{2}, \ldots, )$, and $\downarrow$ denotes falling factorial power $x^{\downarrow m}:=x(x-1)\cdots(x-m+1)$ if $m=1,2,\ldots,$ and $x^{\downarrow 0}=1$ for $m=0$. The relation Eq.\eqref{eq:Ratio-skewSYT-shiftedSchur} admits asymptotic expression by graded symmetric algebra \cite{Borodin_Olshanski_2016}.

Shifted Schur function can be computed by the following combinatorial formula \cite{Goulden1994superSchur,okounkov1996shifted}:
\begin{align}
s_{\mu}^{*}(x_{1},\ldots,)
=
\sum_{T} \sum_{\Box \in \mu} (x_{T(\Box)}-c(\Box)),
\end{align}
where the sum runs over reverse semi-standard Young tableaux $T$ and if $\Box=(i,j)$ then $c(\Box)=j-i$ (called content). Using this formula, the ratio is given by as follows,
\begin{align}
\frac{
\dim \mathbb{Y}_{\lambda / (1^{k})}
}{
\dim \mathbb{Y}_{\lambda / (2,1^{k-2})}
}
&=
\frac{1}{\frac{N-\lambda_{k}}{\lambda_{k}} \prod_{i=1}^{k-1}(1-\frac{1}{\lambda_{i}+k-i} )
+
\sum_{j=1}^{k-1} \frac{N-1}{\lambda_{j}+k-j}
\prod_{i=1}^{j-1} (1-\frac{1}{\lambda_{i}+k-i})
 }
\end{align}
For example, this gives ratio $1/3$ for $k=3$ and $N=4$ under $\lambda=(3,2,1)$.

Denote $\omega=(\omega_{1},\ldots,\omega_{k})$ with $\omega_{i}=\lambda_{i}/(N+k-1)$ for $i=1,\ldots,k$. Asymptotically, we have
\begin{align}
\frac{
\dim \mathbb{Y}_{\lambda / (1^{k})}
}{
\dim \mathbb{Y}_{\lambda / (2,1^{k-2})}
}
\overset{N \to \infty}{\sim}
\frac{
1
}{
\sum_{j=1}^{k} \frac{1}{\omega_{j}} - 1 }
\leq
\frac{1}{k^2-1},
\end{align}
since the harmonic mean is lower or equal to the arithmetic mean. The equality holds if and only if $\omega_{1}=\cdots=\omega_{k}=1/k$ corresponding to the rectangular shape of Young diagrams.

Note that $\mathbb{Y}_{\lambda / (1^{k})}$ is an irreducible representation of $S_{N-1}$ corresponding to $(\lambda_{1}-1,\ldots,\lambda_{k}-1)$. Asymptotically, the $\dim \mathbb{Y}_{\lambda / (1^{k})}$ could be expressed in big O notation (e.g., by using formula, Proposition 2.1 in \cite{regev2010asymptSYT}),
\begin{align}
\dim \mathbb{Y}_{\lambda / (1^{k})}
\sim
O(k^{N-1} (N-1)^{-\frac{k^2+k-2}{4}}).
\end{align}

%%%%%%%%

%------------------------------------------------

%%%%%%%%

\section{Conclusion \& Perspectives}

In this paper, we studied the problem of testing $k$-block-positivity through the $k$-extension and extendibility hierarchy. The $k$-extension converts the problem of $k$-block-positivity testing into block-positivity testing, thus providing a mathematical starting point towards computational testing. Considering the large computational resources this may require, symmetry reduction is considered. We use dualization that converts $\bar{U} \otimes U$-symmetry to $U \otimes U$-symmetry for two reasons.
(i) the symmetry reduction relies on Schur-Weyl duality, which allows us to decompose feasible states into the Young diagram-blocks; (ii) the nonzero objective function is supported on Young diagrams having $k$-rows. We then show how to implement the permutational invariance on the diagram-blocks, and estimate the size of the diagram-blocks. Within a $\lambda$-block, there are two classes of tableau-labeled matrices: (a) matrices contribute to the objective function $\tr X \rho$ that are labeled by standard Young tableaux of $\lambda / (1^{k})$; (b) matrices balance the trace $\tr \rho=1$ via permutational constraints. The ratio of the numbers of the two classes of matrices, relates to the shifted Schur functions with respect to $\lambda$, and the maximal ratio takes when $\lambda$ is rectangle.

%%%%%%%%

%------------------------------------------------

%%%%%%%%

\section*{Acknowledgements}

Q.C. and O.F. acknowledge financial support from the European Research Council under the Grant Agreement No.~851716 (ERC, AlgoQIP) and from the European Union’s Horizon 2020 research and innovation programme under Grant Agreement No 101017733 (QuantERA II, VERIqTAS).
B. C. is supported by JSPS Grant-in-Aid Scientific Research (B) no. 21H00987, and Challenging Research (Exploratory) no.  23K17299.
B. C. was also supported by the Chaire Jean Morlet and acknowledges the hospitality of ENS Lyon during the fall 2024, which allowed us to initiate this project. 

%%%%%%%%

%------------------------------------------------

%%%%%%%%

\begin{appendices}

%%%%%%%%

\section{Example: isotropic states at $k=2$} \label{app:example-k=2}

This appendix is meant to provide a self-contained illustration for the context of the paper through the example of the isotropic state with parameter $\alpha$ for the case of $k=2$,
\begin{align}
X=\id_{d} \otimes \id_{d} + \alpha d \ketbra{\phi_{d}}{\phi_{d}}.
\end{align}
Beginning with level $N=1$, we write the maximally entangled state for the auxiliary system as $\ket{\phi_2} = \frac{1}{\sqrt{2}}(\ket{1^*1} + \ket{2^*2})$. Under the change of basis $\ket{1^{*}} \mapsto -\ket{2}$ and $\ket{2^{*}} \mapsto \ket{1}$, we can write $\ket{\phi_2} = \frac{1}{\sqrt{2}}(\ket{12} - \ket{21})$. More generally, for an arbitrary unitary $U$, we have
\begin{align}
\bar{U} \otimes U (\ket{1^{*} 1} + \ket{2^{*} 2})
=
\frac{U \otimes U}{\det U} (\ket{1 2} - \ket{2 1}).
\end{align}
Under this change of basis,
\begin{align}
\ketbra{\phi_{2}}{\phi_{2}} \otimes X
=
\Pi_{2} \otimes X=
\frac{1}{2}
\begin{pmatrix}
0 & 0 & 0 & 0 \\
0 & X & -X & 0 \\
0 & -X & X & 0 \\
0 & 0 & 0 & 0
\end{pmatrix},
\end{align}
where $\Pi_{2}= \frac{1}{2} (\ket{1 2} - \ket{2 1})(\bra{1 2} - \bra{2 1})$. The factor $\frac{1}{2}$ is the normalization factor of $\ket{\phi_{2}}$.
At level $N=1$ the Schur transform $T$ (recall Section~\ref{sec:schur}) changes the basis $\{ \ket{11} , \ket{12} , \ket{21} , \ket{22} \}$ to Schur basis $\ket{p_{\lambda} , q_{\lambda} }$ as below
\begin{align}
\begin{pmatrix}
\ket{\ytableausetup{boxsize=0.7em} \begin{ytableau} {\scriptstyle 1} \\ {\scriptstyle 2} \end{ytableau} , \begin{ytableau} {\scriptstyle 1} \\ {\scriptstyle 2} \end{ytableau}} \\
\ket{\begin{ytableau} {\scriptstyle 1} & {\scriptstyle 2} \end{ytableau} , \begin{ytableau} {\scriptstyle 1} & {\scriptstyle 1} \end{ytableau}} \\
\ket{\begin{ytableau} {\scriptstyle 1} & {\scriptstyle 2} \end{ytableau} , \begin{ytableau} {\scriptstyle 1} & {\scriptstyle 2} \end{ytableau}} \\
\ket{\begin{ytableau} {\scriptstyle 1} & {\scriptstyle 2} \end{ytableau} , \begin{ytableau} {\scriptstyle 2} & {\scriptstyle 2} \end{ytableau}}
\end{pmatrix}
=
\underbrace{
\begin{pmatrix}
0 & \frac{1}{\sqrt{2}} & -\frac{1}{\sqrt{2}} & 0 \\
1 & 0 & 0 & 0 \\
0 & \frac{1}{\sqrt{2}} & \frac{1}{\sqrt{2}} & 0 \\
0 & 0 & 0 & 1
\end{pmatrix}
}_{T}
\begin{pmatrix}
\ket{11} \\
\ket{12} \\
\ket{21} \\
\ket{22}
\end{pmatrix},
\end{align}
where the first slot for $\ket{p_{\lambda}}$ of the Schur basis, for example $\ket{\begin{ytableau} {\scriptstyle 1} & {\scriptstyle 2} \end{ytableau} , \begin{ytableau} {\scriptstyle 2} & {\scriptstyle 2} \end{ytableau}}$, is standard Young tableau and the second slot for $\ket{q_{\lambda}}$ semistandard Young tableau.
The $k$-extension $\Pi_{2} \otimes X$ is diagonalized by Schur transform, i.e., in the Schur basis, we have
\begin{align}
\Pi_{2} \otimes X
=
\begin{pmatrix}
X & 0 & 0 & 0 \\
0 & 0 & 0 & 0 \\
0 & 0 & 0 & 0 \\
0 & 0 & 0 & 0
\end{pmatrix}
=
\ketbra{ \begin{ytableau} {\scriptstyle 1} \\ {\scriptstyle 2} \end{ytableau} }{ \begin{ytableau} {\scriptstyle 1} \\ {\scriptstyle 2} \end{ytableau} }
\otimes
\ketbra{ \begin{ytableau} {\scriptstyle 1} \\ {\scriptstyle 2} \end{ytableau} }{ \begin{ytableau} {\scriptstyle 1} \\ {\scriptstyle 2} \end{ytableau} }
\otimes X.
\end{align}
Now for $\rho$ that is invariant under $U \otimes U$ on the auxiliary $k$-extension space, it can be shown\footnote{Beside recalling Theorem~\ref{thm:Twirling-PartialUnitary}, the form can be verified by using Weingarten calculus \cite{Collins2006Weingarten} and Schur transform.
}
that it has the following form
\begin{align}
\rho
&=
\begin{pmatrix}
w_{\begin{ytableau} \quad \\ \quad \end{ytableau}} \rho_{ \begin{ytableau} {\scriptstyle 1} \\ {\scriptstyle 2} \end{ytableau} , \begin{ytableau} {\scriptstyle 1} \\ {\scriptstyle 2} \end{ytableau} } & 0 & 0 & 0 \\
0 & \frac{1}{3} w_{\begin{ytableau} \quad & \quad \end{ytableau}}  \rho_{ \begin{ytableau} {\scriptstyle 1} & {\scriptstyle 2} \end{ytableau} , \begin{ytableau} {\scriptstyle 1} & {\scriptstyle 2} \end{ytableau} } & 0 & 0 \\
0 & 0 &  \frac{1}{3} w_{\begin{ytableau} \quad & \quad \end{ytableau}} \rho_{ \begin{ytableau} {\scriptstyle 1} & {\scriptstyle 2} \end{ytableau} , \begin{ytableau} {\scriptstyle 1} & {\scriptstyle 2} \end{ytableau} } & 0 \\
0 & 0 & 0 &  \frac{1}{3} w_{\begin{ytableau} \quad & \quad \end{ytableau}} \rho_{ \begin{ytableau} {\scriptstyle 1} & {\scriptstyle 2} \end{ytableau} , \begin{ytableau} {\scriptstyle 1} & {\scriptstyle 2} \end{ytableau} }
\end{pmatrix},
\end{align}
where $w_{\begin{ytableau} \quad \\ \quad \end{ytableau}}$ and $w_{\begin{ytableau} \quad & \quad \end{ytableau}}$ are nonnegative numbers associated with the Young diagrams $\begin{ytableau} \quad \\ \quad \end{ytableau}$ and $\begin{ytableau} \quad & \quad \end{ytableau}$ that satisfy $w_{\begin{ytableau} \quad \\ \quad \end{ytableau}}+w_{\begin{ytableau} \quad & \quad \end{ytableau}}=1$ due to the trace condition $\tr \rho=1$.
The $\rho_{ \begin{ytableau} {\scriptstyle 1} \\ {\scriptstyle 2} \end{ytableau} } , \rho_{ \begin{ytableau} {\scriptstyle 1} & {\scriptstyle 2} \end{ytableau} , \begin{ytableau} {\scriptstyle 1} & {\scriptstyle 2} \end{ytableau} } \in \mathbb{M}_{d^2 \times d^2}(\mathbb{C})$ are tableau-labeled matrices having trace one. In conclusion, we can write the reduced SDP as follows:
\begin{align}
\mathsf{SDP}_{k,1}^{\sym}(X)
&:=
\min_{\rho \geq 0} 
\tr [ ( \Pi_{2} \otimes X ) \rho ]
=
\min_{ \rho_{ \begin{ytableau} {\scriptstyle 1} \\ {\scriptstyle 2} \end{ytableau} , \begin{ytableau} {\scriptstyle 1} \\ {\scriptstyle 2} \end{ytableau} } \geq 0}
\tr [ X \rho_{ \begin{ytableau} {\scriptstyle 1} \\ {\scriptstyle 2} \end{ytableau} , \begin{ytableau} {\scriptstyle 1} \\ {\scriptstyle 2} \end{ytableau} } ]
, \\
\text{subject to } \:
&
\tr\rho_{ \begin{ytableau} {\scriptstyle 1} \\ {\scriptstyle 2} \end{ytableau} , \begin{ytableau} {\scriptstyle 1} \\ {\scriptstyle 2} \end{ytableau} }=1,
\nonumber
\end{align}
where we have set $w_{\begin{ytableau} \quad \\ \quad \end{ytableau}}=1$ and $w_{\begin{ytableau} \quad & \quad \end{ytableau}}=0$, since $\tr[(\Pi_{2} \otimes X) \rho_{\begin{ytableau} \quad & \quad \end{ytableau}}]=0$, setting $w_{\begin{ytableau} \quad & \quad \end{ytableau}}=0$ could maximize the negativity.

\bigskip

At level $N=2$, Schur transform $T$ changes the basis $\{ \ket{111} , \ket{112} , \ldots \}$ to Schur basis, read as
\begin{align}
\begin{pmatrix}
\ket{\begin{ytableau} {\scriptstyle 1} & {\scriptstyle 3} \\ {\scriptstyle 2} \end{ytableau} , \begin{ytableau} {\scriptstyle 1} & {\scriptstyle 1} \\ {\scriptstyle 2} \end{ytableau}} \\
\ket{\begin{ytableau} {\scriptstyle 1} & {\scriptstyle 3} \\ {\scriptstyle 2} \end{ytableau} , \begin{ytableau} {\scriptstyle 1} & {\scriptstyle 2} \\ {\scriptstyle 2} \end{ytableau} } \\
\ket{\begin{ytableau} {\scriptstyle 1} & {\scriptstyle 2} \\ {\scriptstyle 3} \end{ytableau} , \begin{ytableau} {\scriptstyle 1} & {\scriptstyle 1} \\ {\scriptstyle 2} \end{ytableau}} \\
\ket{\begin{ytableau} {\scriptstyle 1} & {\scriptstyle 2} \\ {\scriptstyle 3} \end{ytableau} , \begin{ytableau} {\scriptstyle 1} & {\scriptstyle 2} \\ {\scriptstyle 2} \end{ytableau}} \\
\ket{\begin{ytableau} {\scriptstyle 1} & {\scriptstyle 2} & {\scriptstyle 3} \end{ytableau} , \begin{ytableau} {\scriptstyle 1} & {\scriptstyle 1} & {\scriptstyle 1} \end{ytableau}} \\
\ket{\begin{ytableau} {\scriptstyle 1} & {\scriptstyle 2} & {\scriptstyle 3} \end{ytableau} , \begin{ytableau} {\scriptstyle 1} & {\scriptstyle 1} & {\scriptstyle 2} \end{ytableau}} \\
\ket{\begin{ytableau} {\scriptstyle 1} & {\scriptstyle 2} & {\scriptstyle 3} \end{ytableau} , \begin{ytableau} {\scriptstyle 1} & {\scriptstyle 2} & {\scriptstyle 2} \end{ytableau}} \\
\ket{\begin{ytableau} {\scriptstyle 1} & {\scriptstyle 2} & {\scriptstyle 3} \end{ytableau} , \begin{ytableau} {\scriptstyle 2} & {\scriptstyle 2} & {\scriptstyle 2} \end{ytableau}}
\end{pmatrix}
=
\underbrace{
\begin{pmatrix}
0 & 0 & \frac{1}{\sqrt{2}} & 0 & -\frac{1}{\sqrt{2}} & 0 & 0 & 0 \\
0 & 0 & 0 & \frac{1}{\sqrt{2}} & 0 & -\frac{1}{\sqrt{2}} & 0 & 0 \\
0 & \frac{2}{\sqrt{6}} & -\frac{1}{\sqrt{6}} & 0 & -\frac{1}{\sqrt{6}} & 0 & 0 & 0 \\
0 & 0 & 0 & \frac{1}{\sqrt{6}} & 0 & \frac{1}{\sqrt{6}} & -\frac{2}{\sqrt{6}} & 0 \\
1 & 0 & 0 & 0 & 0 & 0 & 0 & 0 \\
0 & \frac{1}{\sqrt{3}} & \frac{1}{\sqrt{3}} & 0 & \frac{1}{\sqrt{3}} & 0 & 0 & 0 \\
0 & 0 & 0 & \frac{1}{\sqrt{3}} & 0 & \frac{1}{\sqrt{3}} & \frac{1}{\sqrt{3}} & 0 \\
0 & 0 & 0 & 0 & 0 & 0 & 0 & 1
\end{pmatrix}
}_{T}
\begin{pmatrix}
\ket{111} \\
\ket{112} \\
\ket{121} \\
\ket{122} \\
\ket{211} \\
\ket{212} \\
\ket{221} \\
\ket{222}
\end{pmatrix}.
\end{align}
For $\rho_{2,2}$ the feasible state of $k=2$ and $N=2$ as Eq.\eqref{eq:Blocks-Nextension-states-dual} that is invariant under $U^{\otimes 3}$ on the auxiliary $k$-extension space, it can be shown that
\begin{align}
&
\rho_{2,2}=
w_{\begin{ytableau} {\scriptstyle \ } & {\scriptstyle \ } \\ {\scriptstyle \ } \end{ytableau}} \left(
\frac{\id_{2}}{2} \otimes
\rho_{\begin{ytableau} {\scriptstyle \ } & {\scriptstyle \ } \\ {\scriptstyle \ } \end{ytableau}}
\right)
\oplus
w_{\begin{ytableau} {\scriptstyle \ } & {\scriptstyle \ } & {\scriptstyle \ } \end{ytableau}} \left( \frac{\id_{4}}{4} \otimes \rho_{\begin{ytableau} {\scriptstyle \ } & {\scriptstyle \ } & {\scriptstyle \ } \end{ytableau}}
\right),
\label{eq:rho,N=2,Schurbasis}
\\
&
\text{where} \
\rho_{\begin{ytableau} {\scriptstyle \ } & {\scriptstyle \ } \\ {\scriptstyle \ } \end{ytableau}}
=
\begin{pmatrix}
a_{11} & a_{12} \\
a_{21} & a_{22}
\end{pmatrix}, \ 
\rho_{\begin{ytableau} {\scriptstyle \ } & {\scriptstyle \ } & {\scriptstyle \ } \end{ytableau}}
=
b.
\end{align}
Notice here that we adopt $\ket{p_{\lambda} q_{\lambda}} \cong \ket{q_{\lambda} p_{\lambda}}$ for getting compact expression, as mentioned in Theorem~\ref{thm:Twirling-PartialUnitary}.
The second line shows diagram-blocks, in which tableau-labeled matrices $a_{11}, \ldots , a_{22} , b \in \mathbb{M}_{d^3 \times d^3}(\mathbb{C})$ are contained, where the labels $1,2$ stand for standard Young tableaux $\begin{ytableau} {\scriptstyle 1} & {\scriptstyle 3} \\ {\scriptstyle 2} \end{ytableau}$ and $\begin{ytableau} {\scriptstyle 1} & {\scriptstyle 2} \\ {\scriptstyle 3} \end{ytableau}$ respectively. Conditions $\tr \rho_{2,2}=1$ and $\rho_{2,2} \geq 0$ lead to
\begin{align}
&
\tr \rho_{\begin{ytableau} {\scriptstyle \ } & {\scriptstyle \ } \\ {\scriptstyle \ } \end{ytableau}}=1, \quad
\tr \rho_{\begin{ytableau} {\scriptstyle \ } & {\scriptstyle \ } & {\scriptstyle \ } \end{ytableau}}=1, \quad
\rho_{\begin{ytableau} {\scriptstyle \ } & {\scriptstyle \ } \\ {\scriptstyle \ } \end{ytableau}}
\geq 0, \quad
\rho_{\begin{ytableau} {\scriptstyle \ } & {\scriptstyle \ } & {\scriptstyle \ } \end{ytableau}} \geq 0.
\end{align}
Note that $\rho_{\begin{ytableau} {\scriptstyle \ } & {\scriptstyle \ } \\ {\scriptstyle \ } \end{ytableau}} \geq 0$ implies $a_{21}=a_{12}^{\dagger}$.
The nonnegative numbers $w_{\begin{ytableau} {\scriptstyle \ } & {\scriptstyle \ } \\ {\scriptstyle \ } \end{ytableau}}, w_{\begin{ytableau} {\scriptstyle \ } & {\scriptstyle \ } & {\scriptstyle \ } \end{ytableau}}$ satisfy $w_{\begin{ytableau} {\scriptstyle \ } & {\scriptstyle \ } \\ {\scriptstyle \ } \end{ytableau}} + w_{\begin{ytableau} {\scriptstyle \ } & {\scriptstyle \ } & {\scriptstyle \ } \end{ytableau}}=1$ due to $\tr \rho_{2,2}=1$. On the other hand, in the Schur basis, the $X_{2,2}=\Pi_{2} \otimes \id_{2} \otimes X \otimes \id_{d}$, which is the $N=2$ level $k$-extension of $X$, is read as
\begin{align}
X_{2,2}
=
\id_{2} \otimes
\ketbra{ \begin{ytableau} {\scriptstyle 1} & {\scriptstyle 3} \\ {\scriptstyle 2} \end{ytableau} }{ \begin{ytableau} {\scriptstyle 1} & {\scriptstyle 3} \\ {\scriptstyle 2} \end{ytableau} }
\otimes (X \otimes \id_{d}).
\label{eq:XtoBeTest,N=2,IsotropicState}
\end{align}
Note that $\ketbra{ \begin{ytableau} {\scriptstyle 1} & {\scriptstyle 3} \\ {\scriptstyle 2} \end{ytableau} }{ \begin{ytableau} {\scriptstyle 1} & {\scriptstyle 3} \\ {\scriptstyle 2} \end{ytableau} }=\mathbb{P}_{\begin{ytableau} {\scriptstyle \ } & {\scriptstyle \ } \\ {\scriptstyle \ } \end{ytableau} / (1^2)}$, which is the projector of skew Young representation $\begin{ytableau} {\scriptstyle \ } & {\scriptstyle \ } \\ {\scriptstyle \ } \end{ytableau} / (1^2) \equiv \begin{ytableau} {\scriptstyle \bullet} & {\scriptstyle \ } \\ {\scriptstyle \bullet} \end{ytableau}$.
This also shows $X_{\begin{ytableau} {\scriptstyle \ } & {\scriptstyle \ } \\ {\scriptstyle \ } \end{ytableau}}$ by
\begin{align}
X_{\begin{ytableau} {\scriptstyle \ } & {\scriptstyle \ } \\ {\scriptstyle \ } \end{ytableau}}
=
\ketbra{ \begin{ytableau} {\scriptstyle 1} & {\scriptstyle 3} \\ {\scriptstyle 2} \end{ytableau} }{ \begin{ytableau} {\scriptstyle 1} & {\scriptstyle 3} \\ {\scriptstyle 2} \end{ytableau} }
\otimes (X \otimes \id_{d})
=
\begin{pmatrix}
X \otimes \id_{d} & 0_{d^3} \\
0_{d^3} & 0_{d^3}
\end{pmatrix}.
\end{align}
Multiplying Eq.~\eqref{eq:XtoBeTest,N=2,IsotropicState} with Eq.~\eqref{eq:rho,N=2,Schurbasis} then trace gets objective function $\tr(X_{2,2} \rho_{2,2})
=\tr(X_{\begin{ytableau} {\scriptstyle \ } & {\scriptstyle \ } \\ {\scriptstyle \ } \end{ytableau}} \rho_{\begin{ytableau} {\scriptstyle \ } & {\scriptstyle \ } \\ {\scriptstyle \ } \end{ytableau}})
=\tr[ (X\otimes \id_{d} ) a_{11} ]$.
The diagram-block $\rho_{\begin{ytableau} {\scriptstyle \ } & {\scriptstyle \ } & {\scriptstyle \ } \end{ytableau}}$ makes no contribution to the objective function. Indeed, it is a Young diagram with rows less than $k$. Hence only $\rho_{\begin{ytableau} {\scriptstyle \ } & {\scriptstyle \ } \\ {\scriptstyle \ } \end{ytableau}}$ is to be taken into account. Hence, set $w_{\begin{ytableau} {\scriptstyle \ } & {\scriptstyle \ } \\ {\scriptstyle \ } \end{ytableau}}=1$. The size of diagram-block $d_{\begin{ytableau} {\scriptstyle \ } & {\scriptstyle \ } \\ {\scriptstyle \ } \end{ytableau}}=d_{\begin{ytableau} {\scriptstyle \bullet} & {\scriptstyle \ } \\ {\scriptstyle \bullet} \end{ytableau}}+d_{\begin{ytableau} {\scriptstyle \bullet} & {\scriptstyle \bullet} \\ {\scriptstyle \ } \end{ytableau}}=2$.

\medskip

Note that whereas $a_{12}, a_{22}$ cannot be set as zero due to the permutational symmetry. This makes $\tr a_{22} \neq 0$ such that the optimal value gets affected by $a_{12}, a_{22}$. Imposing the permutational symmetry is realized by representing Coxeter generator $\tau_{2} \equiv (2,3)$ into $\Delta_{\begin{ytableau} {\scriptstyle \ } & {\scriptstyle \ } \\ {\scriptstyle \ } \end{ytableau}}$, which can be achieved by tensoring the irreducible representation matrix $(2,3)_{\begin{ytableau} {\scriptstyle \ } & {\scriptstyle \ } \\ {\scriptstyle \ } \end{ytableau}}=\begin{pmatrix}
\frac{1}{2} & \frac{\sqrt{3}}{2} \\
\frac{\sqrt{3}}{2} & -\frac{1}{2}
\end{pmatrix}$ with the canonical permutation representation $(2,3)$,
\begin{align}
\Delta_{\begin{ytableau} {\scriptstyle \ } & {\scriptstyle \ } \\ {\scriptstyle \ } \end{ytableau}}((2,3))
=
\begin{pmatrix}
\frac{1}{2} & \frac{\sqrt{3}}{2} \\
\frac{\sqrt{3}}{2} & -\frac{1}{2}
\end{pmatrix}
\otimes
(2,3),
\label{eq:ex-Delta2,1-2,3}
\end{align}
Hence the constraint of the permutational symmetry $\Delta_{\begin{ytableau} {\scriptstyle \ } & {\scriptstyle \ } \\ {\scriptstyle \ } \end{ytableau}}((2,3))\rho_{\begin{ytableau} {\scriptstyle \ } & {\scriptstyle \ } \\ {\scriptstyle \ } \end{ytableau}}=\rho_{\begin{ytableau} {\scriptstyle \ } & {\scriptstyle \ } \\ {\scriptstyle \ } \end{ytableau}}$ is read as
\begin{align}
\begin{pmatrix}
a_{11} & a_{12} \\
a_{12}^{\dagger} & a_{22}
\end{pmatrix}
=
\Delta_{\begin{ytableau} {\scriptstyle \ } & {\scriptstyle \ } \\ {\scriptstyle \ } \end{ytableau}}((2,3))
\begin{pmatrix}
a_{11} & a_{12} \\
a_{12}^{\dagger} & a_{22}
\end{pmatrix}
=
\begin{pmatrix}
\frac{1}{2} \id_{d}^{\otimes 3} & \frac{\sqrt{3}}{2} \id_{d}^{\otimes 3} \\
\frac{\sqrt{3}}{2} \id_{d}^{\otimes 3} & -\frac{1}{2} \id_{d}^{\otimes 3}
\end{pmatrix}
\begin{pmatrix}
(2,3) a_{11} & (2,3) a_{12} \\
(2,3) a_{21} & (2,3) a_{22}
\end{pmatrix}.
\end{align}
Eventually, the level $N=2$ reduced SDP is then shown as follows,
\begin{align}
\mathsf{SDP}_{k,2}^{\sym}(X)
&:=
\min_{\rho_{2,2} \geq 0} 
\tr (X_{2,2} \rho_{2,2} )
=
\min_{ \rho_{\begin{ytableau} {\scriptstyle \ } & {\scriptstyle \ } \\ {\scriptstyle \ } \end{ytableau}} \in \mathrm{Pos}(\mathbb{C}^{2} \otimes (\mathbb{C}^{d})^{\otimes 3}) }
\tr [ (X \otimes \id_{d}) a_{11} ]
, \\
\text{subject to } \:
&
\tr \rho_{\begin{ytableau} {\scriptstyle \ } & {\scriptstyle \ } \\ {\scriptstyle \ } \end{ytableau}}=1,
\ \text{and} \
\Delta_{\begin{ytableau} {\scriptstyle \ } & {\scriptstyle \ } \\ {\scriptstyle \ } \end{ytableau}}((2,3))
\rho_{\begin{ytableau} {\scriptstyle \ } & {\scriptstyle \ } \\ {\scriptstyle \ } \end{ytableau}}
=
\rho_{\begin{ytableau} {\scriptstyle \ } & {\scriptstyle \ } \\ {\scriptstyle \ } \end{ytableau}}
\ \text{where} \
\rho_{\begin{ytableau} {\scriptstyle \ } & {\scriptstyle \ } \\ {\scriptstyle \ } \end{ytableau}}
=
\begin{pmatrix}
a_{11} & a_{12} \\
a_{12}^{\dagger} & a_{22}
\end{pmatrix}.
\nonumber
\end{align}

\bigskip

The $N=3$ level is done as the same manner. In Schur basis $X_{2,3}=\Pi_{2} \otimes \id_{2}^{\otimes 2} \otimes X \otimes \id_{d}^{\otimes 2}$ is,
\begin{align}
X_{2,3}
=
\id_{1} \otimes
\ketbra{ \begin{ytableau} {\scriptstyle 1} & {\scriptstyle 3} \\ {\scriptstyle 2} & {\scriptstyle 4} \end{ytableau} }{ \begin{ytableau} {\scriptstyle 1} & {\scriptstyle 3} \\ {\scriptstyle 2} & {\scriptstyle 4} \end{ytableau} }
\otimes (X \otimes \id_{d}^{\otimes 2})
\oplus
\id_{3} \otimes
\ketbra{ \begin{ytableau} {\scriptstyle 1} & {\scriptstyle 3} & {\scriptstyle 4} \\ {\scriptstyle 2} \end{ytableau} }{ \begin{ytableau} {\scriptstyle 1} & {\scriptstyle 3} & {\scriptstyle 4} \\ {\scriptstyle 2} \end{ytableau} }
\otimes (X \otimes \id_{d}^{\otimes 2})
.
\label{eq:XtoBeTest,N=3,IsotropicState}
\end{align}
The $\ketbra{ \begin{ytableau} {\scriptstyle 1} & {\scriptstyle 3} \\ {\scriptstyle 2} & {\scriptstyle 4} \end{ytableau} }{ \begin{ytableau} {\scriptstyle 1} & {\scriptstyle 3} \\ {\scriptstyle 2} & {\scriptstyle 4} \end{ytableau} }=\mathbb{P}_{\begin{ytableau} {\scriptstyle \ } & {\scriptstyle \ } \\ {\scriptstyle \ } & {\scriptstyle \ } \end{ytableau} / (1^2)}$ gives the projector of skew representation $\begin{ytableau} {\scriptstyle \ } & {\scriptstyle \ } \\ {\scriptstyle \ } & {\scriptstyle \ } \end{ytableau} / (1^2) \equiv \begin{ytableau} {\scriptstyle \bullet} & {\scriptstyle \ } \\ {\scriptstyle \bullet} & {\scriptstyle \ } \end{ytableau}$. Likewise, the $\ketbra{ \begin{ytableau} {\scriptstyle 1} & {\scriptstyle 3} & {\scriptstyle 4} \\ {\scriptstyle 2} \end{ytableau} }{ \begin{ytableau} {\scriptstyle 1} & {\scriptstyle 3} & {\scriptstyle 4} \\ {\scriptstyle 2} \end{ytableau} }=\mathbb{P}_{\begin{ytableau} {\scriptstyle \ } & {\scriptstyle \ } & {\scriptstyle \ } \\ {\scriptstyle \ } \end{ytableau} / (1^2)}$ gives the projector of skew representation $\begin{ytableau} {\scriptstyle \ } & {\scriptstyle \ } & {\scriptstyle \ } \\ {\scriptstyle \ } \end{ytableau} / (1^2) \equiv \begin{ytableau} {\scriptstyle \bullet} & {\scriptstyle \ } & {\scriptstyle \ } \\ {\scriptstyle \bullet} \end{ytableau}$.
Other relevant skew shapes are $\begin{ytableau} {\scriptstyle \ } & {\scriptstyle \ } \\ {\scriptstyle \ } & {\scriptstyle \ } \end{ytableau} / (2) \equiv \begin{ytableau} {\scriptstyle \bullet} & {\scriptstyle \bullet} \\ {\scriptstyle \ } & {\scriptstyle \ } \end{ytableau}$ and $\begin{ytableau} {\scriptstyle \ } & {\scriptstyle \ } & {\scriptstyle \ } \\ {\scriptstyle \ } \end{ytableau} / (2) \equiv \begin{ytableau} {\scriptstyle \bullet} & {\scriptstyle \bullet} & {\scriptstyle \ } \\ {\scriptstyle \ } \end{ytableau}$.

For $\rho_{2,3}$ that is invariant under $U^{\otimes 4}$ on the auxiliary $k$-extension space, it can be shown that
\begin{align}
&
\rho_{2,3}
=
w_{\begin{ytableau} {\scriptstyle \ } & {\scriptstyle \ } \\ {\scriptstyle \ } & {\scriptstyle \ } \end{ytableau}} \left( \id_{1} \otimes
\rho_{\begin{ytableau} {\scriptstyle \ } & {\scriptstyle \ } \\ {\scriptstyle \ } & {\scriptstyle \ } \end{ytableau}}
\right)
\oplus
w_{\begin{ytableau} {\scriptstyle \ } & {\scriptstyle \ } & {\scriptstyle \ } \\ {\scriptstyle \ } \end{ytableau}}
\left( \frac{\id_{3}}{3} \otimes \rho_{\begin{ytableau} {\scriptstyle \ } & {\scriptstyle \ } & {\scriptstyle \ } \\ {\scriptstyle \ } \end{ytableau}}
\right)
\oplus
w_{\begin{ytableau} {\scriptstyle \ } & {\scriptstyle \ } & {\scriptstyle \ } & {\scriptstyle \ } \end{ytableau}}
\left( \frac{\id_{5}}{5} \otimes \rho_{\begin{ytableau} {\scriptstyle \ } & {\scriptstyle \ } & {\scriptstyle \ } & {\scriptstyle \ } \end{ytableau}} \right), \\
&
\text{where} \
\rho_{\begin{ytableau} {\scriptstyle \ } & {\scriptstyle \ } \\ {\scriptstyle \ } & {\scriptstyle \ } \end{ytableau}}
=
\begin{pmatrix}
a_{11} & a_{12} \\
a_{12}^{\dagger} & a_{22}
\end{pmatrix},
\quad
\rho_{\begin{ytableau} {\scriptstyle \ } & {\scriptstyle \ } & {\scriptstyle \ } \\ {\scriptstyle \ } \end{ytableau}}
=
\begin{pmatrix}
b_{11} & b_{12} & b_{13} \\
b_{12}^{\dagger} & b_{22} & b_{23} \\
b_{13}^{\dagger} & b_{22}^{\dagger} & b_{33}
\end{pmatrix},
\rho_{\begin{ytableau} {\scriptstyle \ } & {\scriptstyle \ } & {\scriptstyle \ } & {\scriptstyle \ } \end{ytableau}}
=
c.
\end{align}
The second line shows diagram-blocks, in which tableau-labeled matrices $a_{11} , \ldots, a_{22} , b_{11} , \ldots , b_{33} , c \in \mathbb{M}_{d^4 \times d^4}(\mathbb{C})$ are contained. The objective function $\tr (X_{2,3} \rho_{2,3})$ is then read as
\begin{align}
&
\tr (X_{2,3} \rho_{2,3})
=
w_{\begin{ytableau} {\scriptstyle \ } & {\scriptstyle \ } \\ {\scriptstyle \ } & {\scriptstyle \ } \end{ytableau}}
\tr \left[ (X \otimes \id_{d}^{\otimes 2}) a_{11} \right]
+
w_{\begin{ytableau} {\scriptstyle \ } & {\scriptstyle \ } & {\scriptstyle \ } \\ {\scriptstyle \ } \end{ytableau}}
\tr \left[ (X \otimes \id_{d}^{\otimes 2}) b_{11} \right].
\end{align}
Again, we can set $w_{\begin{ytableau} {\scriptstyle \ } & {\scriptstyle \ } & {\scriptstyle \ } & {\scriptstyle \ } \end{ytableau}}=0$ since Young diagram $\begin{ytableau} {\scriptstyle \ } & {\scriptstyle \ } & {\scriptstyle \ } & {\scriptstyle \ } \end{ytableau}$ makes no contribution to the objective function. Indeed, it is a Young diagram with rows less than $k$.
The sizes of diagram-blocks are $d_{\begin{ytableau} {\scriptstyle \ } & {\scriptstyle \ } \\ {\scriptstyle \ } & {\scriptstyle \ } \end{ytableau}}=d_{\begin{ytableau} {\scriptstyle \bullet} & {\scriptstyle \ } \\ {\scriptstyle \bullet} & {\scriptstyle \ } \end{ytableau}}+d_{\begin{ytableau} {\scriptstyle \bullet} & {\scriptstyle \bullet} \\ {\scriptstyle \ } & {\scriptstyle \ } \end{ytableau}}=2$, $d_{\begin{ytableau} {\scriptstyle \ } & {\scriptstyle \ } & {\scriptstyle \ } \\ {\scriptstyle \ } \end{ytableau}}=d_{\begin{ytableau} {\scriptstyle \bullet} & {\scriptstyle \ } & {\scriptstyle \ } \\ {\scriptstyle \bullet} \end{ytableau}} + d_{\begin{ytableau} {\scriptstyle \bullet} & {\scriptstyle \bullet} & {\scriptstyle \ } \\ {\scriptstyle \ } \end{ytableau}}=3$.

\medskip

Only $a_{11}$ and $b_{11}$ make contribution to the objective function, while $a_{12} , a_{22} , b_{12} , \ldots , b_{33}$ affect the optimal value by permutation constraint. Imposing the permutational symmetry is realized by representing Coxeter generators $\tau_{2} \equiv (2,3)$ and $\tau_{3} \equiv (3,4)$ into $\Delta_{\begin{ytableau} {\scriptstyle \ } & {\scriptstyle \ } \\ {\scriptstyle \ } & {\scriptstyle \ } \end{ytableau}}$ and $\Delta_{\begin{ytableau} {\scriptstyle \ } & {\scriptstyle \ } & {\scriptstyle \ } \\ {\scriptstyle \ } \end{ytableau}}$ as follows,
\begin{align}
&
\Delta_{\begin{ytableau} {\scriptstyle \ } & {\scriptstyle \ } \\ {\scriptstyle \ } & {\scriptstyle \ } \end{ytableau}}((2,3))
=
\begin{pmatrix}
\frac{1}{2} & \frac{\sqrt{3}}{2} \\
\frac{\sqrt{3}}{2} & -\frac{1}{2}
\end{pmatrix}
\otimes
(2,3), \quad
\Delta_{\begin{ytableau} {\scriptstyle \ } & {\scriptstyle \ } \\ {\scriptstyle \ } & {\scriptstyle \ } \end{ytableau}}((3,4))
=
\begin{pmatrix}
-1 & 0 \\
0 & 1
\end{pmatrix}
\otimes
(3,4),
\label{eq:ex-Delta2,2}
\\
&
\Delta_{\begin{ytableau} {\scriptstyle \ } & {\scriptstyle \ } & {\scriptstyle \ } \\ {\scriptstyle \ } \end{ytableau}} ((2,3))
=
\begin{pmatrix}
\frac{1}{2} & \frac{\sqrt{3}}{2} & 0 \\
\frac{\sqrt{3}}{2} & -\frac{1}{2} & 0 \\
0 & 0 & 1
\end{pmatrix}
\otimes (2,3), \quad
\Delta_{\begin{ytableau} {\scriptstyle \ } & {\scriptstyle \ } & {\scriptstyle \ } \\ {\scriptstyle \ } \end{ytableau}} ((3,4))
=
\begin{pmatrix}
1 & 0 & 0 \\
0 & \frac{1}{3} & \frac{2\sqrt{2}}{3} \\
0 & \frac{2\sqrt{2}}{3} & -\frac{1}{3}
\end{pmatrix}
\otimes (3,4),
\label{eq:ex-Delta3,1}
\end{align}
which give the the permutational constraints
\begin{align*}
\begin{pmatrix}
a_{11} & a_{12} \\
a_{12}^{\dagger} & a_{22}
\end{pmatrix}
&=
\begin{pmatrix}
\frac{1}{2} \id_{d}^{\otimes 3} & \frac{\sqrt{3}}{2} \id_{d}^{\otimes 3} \\
\frac{\sqrt{3}}{2} \id_{d}^{\otimes 3} & -\frac{1}{2} \id_{d}^{\otimes 3}
\end{pmatrix}
\begin{pmatrix}
(2,3) a_{11} & (2,3) a_{12} \\
(2,3) a_{12}^{\dagger} & (2,3) a_{22}
\end{pmatrix}
\\
&=
\begin{pmatrix}
-\id_{d}^{\otimes 3} & 0 \\
0 & \id_{d}^{\otimes 3}
\end{pmatrix}
\begin{pmatrix}
(3,4) a_{11} & (3,4) a_{12} \\
(3,4) a_{12}^{\dagger} & (3,4) a_{22}
\end{pmatrix},
\\
\begin{pmatrix}
b_{11} & b_{12} & b_{13} \\
b_{12}^{\dagger} & b_{22} & b_{23} \\
b_{13}^{\dagger} & b_{23}{\dagger}  & b_{33}
\end{pmatrix}
&=
\begin{pmatrix}
\frac{1}{2} \id_{d}^{\otimes 3} & \frac{\sqrt{3}}{2} \id_{d}^{\otimes 3} & 0 \\
\frac{\sqrt{3}}{2} \id_{d}^{\otimes 3} & -\frac{1}{2} \id_{d}^{\otimes 3} & 0 \\
0 & 0 & \id_{d}^{\otimes 3}
\end{pmatrix}
\begin{pmatrix}
(2,3) b_{11} & (2,3) b_{12} & (2,3) b_{13} \\
(2,3) b_{12}^{\dagger} & (2,3) b_{22} & (2,3) b_{23} \\
(2,3) b_{13}^{\dagger} & (2,3) b_{23}^{\dagger} & (2,3) b_{33} \\
\end{pmatrix}
\\
&=
\begin{pmatrix}
\id_{d}^{\otimes 3} & 0 & 0 \\
0 & \frac{1}{3} \id_{d}^{\otimes 3} & \frac{2\sqrt{2}}{2} \id_{d}^{\otimes 3} \\
0 & \frac{2\sqrt{2}}{3} \id_{d}^{\otimes 3} & -\frac{1}{3} \id_{d}^{\otimes 3}
\end{pmatrix}
\begin{pmatrix}
(3,4) b_{11} & (3,4) b_{12} & (3,4) b_{13} \\
(3,4) b_{12}^{\dagger} & (3,4) b_{22} & (3,4) b_{23} \\
(3,4) b_{13}^{\dagger} & (3,4) b_{23}^{\dagger} & (3,4) b_{33}
\end{pmatrix}.
\end{align*}
Eventually, the level $N=3$ reduced SDP is then shown as follows,
\begin{align}
\mathsf{SDP}_{k,3}^{\sym}(X)
&:=
\min_{ \{ \rho_{\lambda} \in \mathrm{Pos}(\mathbb{C}^{d_{\lambda}} \otimes (\mathbb{C}^{d})^{\otimes 4}) , \lambda=\begin{ytableau} {\scriptstyle \ } & {\scriptstyle \ } \\ {\scriptstyle \ } & {\scriptstyle \ } \end{ytableau} , \begin{ytableau} {\scriptstyle \ } & {\scriptstyle \ } & {\scriptstyle \ } \\ {\scriptstyle \ } \end{ytableau} \}}
\tr [ (\mathbb{P}_{\lambda / (1^2)} \otimes X \otimes \id_{d}) \rho_{\lambda} ]
, \\
\text{subject to } \:
&
\tr \rho_{\lambda}=1, \quad
\Delta_{\lambda}(\tau_{2})
\rho_{\lambda}
=
\rho_{\lambda}, \quad
\Delta_{\lambda}(\tau_{3})
\rho_{\lambda}
=
\rho_{\lambda}.
\nonumber
\end{align}
Here the nonnegative numbers $w_{\begin{ytableau} {\scriptstyle \ } & {\scriptstyle \ } \\ {\scriptstyle \ } & {\scriptstyle \ } \end{ytableau}}$ and $w_{\begin{ytableau} {\scriptstyle \ } & {\scriptstyle \ } & {\scriptstyle \ } \\ {\scriptstyle \ } \end{ytableau}}$ disappear, for the reason that the optimal value can be viewed as convex linear combination with respect to them.

\end{appendices}

%\printbibliography
\bibliographystyle{plain}
\bibliography{references}

\end{document}